\documentclass{jot}

\usepackage{booktabs} 
\usepackage{graphicx}
\usepackage{import}
\usepackage[size=scriptsize]{todonotes}
\usepackage{mathtools}
\usepackage{xcolor}
\usepackage{latexsym}
\usepackage{amsmath}
\usepackage{hyperref}
\usepackage{mathpartir}
\usepackage{multirow}
\usepackage{csquotes}
\usepackage{xspace}
\usepackage{graphicx}
\usepackage{subcaption}
\usepackage{afterpage}
\usepackage{placeins}
\usepackage{amssymb}
\usepackage{textcomp}
\usepackage{comment}
\usepackage{thmtools,thm-restate}
\usepackage{ifthen}
\usepackage{stackengine}
\usepackage{hyperref}
\usepackage[T1]{fontenc}
\usepackage{breakurl}
\usepackage{rotating}
\usepackage{wrapfig}
\newcommand{\qed}{$\hfill\blacksquare$}
\newtheorem{definition}{Definition}
\newtheorem{example}{Example}[section]
\newtheorem{proposition}{Proposition}
\newtheorem{lemma}{Lemma}
\newenvironment{proof}{\noindent\textit{Proof. }}{\qed\\}

\allowdisplaybreaks
\renewcommand{\todo}[1]{}

\newcommand{\enab}{\textit{enab}}
\newcommand{\sink}{\mathit{Sink}}

\newcommand{\id}{\mathit{id}}

\newcommand{\IV}{\mathit{IV}}
\newcommand{\LV}{\mathit{MV}}
\newcommand{\CV}{\mathit{CV}}

\newcommand{\OG}{{\mathit{OG}}}
\newcommand{\IG}{{\mathit{IG}}}
\renewcommand{\L}{{\mathcal{L}}}
\renewcommand{\S}{{\mathcal{S}}}
\newcommand{\B}{{\mathcal{B}}}
\newcommand{\BB}{{\mathbf{B}}}

\newcommand{\OC}{$\stackanchor[1pt]{\Op}{\Cl}$}
\newcommand{\N}{{\ensuremath{\mathsf{N}}\xspace}}

\newcommand{\iv}{\mathit{iv}}
\newcommand{\ivbar}{\bar{\mathit{\imath v}}}

\newcommand{\loc}{\textit{Loc}}
\newcommand{\val}{\textit{Val}}
\renewcommand{\sl}{\mathit{sl}}
\newcommand{\tl}{\mathit{tl}}
\newcommand{\il}{\mathit{il}}

\newcommand{\f}{{\ensuremath{\mathbf{f}}\xspace}}
\newcommand{\Op}{{\ensuremath{\circ}\xspace}}
\newcommand{\Cl}{{\ensuremath{\bullet}\xspace}}
\newcommand{\definedon}{\downarrow}
\newcommand{\undefinedon}{\uparrow}
\newcommand{\upshift}{{\Uparrow}}
\newcommand{\attime}[1]{^{(#1)}}

\newcommand{\step}[1]{\textbf{\textit{#1}}}

\newenvironment{scenarios}[1][]{%
\footnotesize
  \begin{enumerate}[
     series=scenarios,
     resume=scenarios,
     label=\textbf{Scenario \arabic*},
     align=left,
     itemsep=3pt,
     parsep=1pt,
     leftmargin=*,
     labelwidth=3cm,
     itemindent=1.5cm,
     topsep=\medskipamount]%
}{%
  \end{enumerate}
}
\newenvironment{steps}{%
  \newcommand{\Given}{\item\step{Given}\xspace}
  \newcommand{\When}{\item\step{When}\xspace}
  \newcommand{\Then}{\item\step{Then}\xspace}
  \renewcommand{\And}{\\\step{And}\xspace}
  \begin{itemize}[nosep]
}{%
  \end{itemize}
}

\newcommand{\VU}{\mathcal V}
\newcommand{\TU}{\mathcal T}
\newcommand{\UU}{\mathcal U}
\newcommand{\GU}{\mathcal G}
\newcommand{\GUU}{\mathcal{GU}}
\newcommand{\IAU}{\mathcal I}

\newcommand{\bracket}[1]{[\![ {#1} ]\!]}
\let\sem\bracket
\newcommand{\vu}{\vartheta}

\newcommand{\restr}[1]{{\restriction_{#1}}}
\newcommand{\setof}[1]{\{{#1}\}}
\newcommand{\gensetof}[2]{\setof{#1\mid{#2}}}
\newcommand{\tupof}[1]{\langle{#1}\rangle}

\newcommand{\set}[1]{\{{#1}\}}
\newcommand{\tup}[1]{\langle{#1}\rangle}

\newcommand{\trans}[1]{
  \inlinetrans{#1}
}
\newcommand{\inlinetrans}[1]{
  \xrightarrow{\lowerlabel{#1}}
}

\newcommand{\ntrans}[1]{\not\trans{\enspace#1}}

\newcommand{\code}[1]{\ensuremath{ \mathsf{#1}\xspace}}

\newcommand{\compat}{\mathrel{\Lleftarrow\mskip-12mu\Rrightarrow}}
\newcommand{\ncompat}{\mathrel{\Lleftarrow\mskip-12mu\not\Rrightarrow}}
\newcommand{\ssby}{\sqsubseteq}

\newcommand{\Bool}{\code{Bool}}

\newcommand{\lowerlabel}[1]{\smash{\mathchoice
		{\raisebox{-1.5pt}{$\scriptstyle#1$}}
		{\raisebox{-1.5pt}{$\scriptstyle#1$}}
		{\raisebox{-1pt}{$\scriptstyle#1$}}
		{\raisebox{-1pt}{$\scriptscriptstyle#1$}}
}}

\newcommand{\ini}{\mathit{ini}}
\newcommand{\switch}{(sl,\alpha,\phi,a,tl)}
\newcommand{\ststuple}{\tupof{L, V, G,{\rightarrow},il,\ini}}
\newcommand{\bddtstuple}{\tupof{L,V,G,{\rightarrow}, il, IG, OG}}
\newcommand{\bddtstupleind}[1]{\tup{L_{#1},V,G,{\rightarrow}_{#1}, il_{#1}, IG_{#1}, OG_{#1}}}
\newcommand{\bddtstosts}[2]{translateToSTS(#1,#2)}
\newcommand{\isomap}{\mathit{Map}}

\newcommand{\DEF}{\enspace{::=}\enspace}

\newcommand{\arrow}{{\rightarrow}}

\newcommand{\Loc}{L}

\newcommand{\prog}[1]{{\textnormal{\texttt{#1}}}}

\newcommand{\suc}{\textit{succ}}

\newcommand{\True}{\mathbf{true}}
\newcommand{\False}{\mathbf{false}}
\DeclareRobustCommand{\Models}{\vdash}

\newcommand{\var}{\textit{vars}}
\newcommand{\univ}{\mathcal{U}}
\newcommand{\sstate}{\mathit{SymState}}
\newcommand{\cstate}{\mathit{SemState}}

\newcommand{\satK}{\textit{sat}}

\newcommand{\passes}{\mathrel{\mathsf{passes}}}
\newcommand{\fails}{\mathrel{\mathsf{fails}}}

\newboolean{saturated}
\setboolean{saturated}{true}

\newcommand{\TC}{\mathit{TC}}

\newcommand{\sat}[1]{{#1}^\satK}

\newcommand{\ISS}[2]{\text{STS}(#1,#2)}

\newcommand{\seg}[1]{\textit{seg}(#1)}
\newcommand{\Seg}{\textit{seg}}

\newcommand{\ogimp}[2]{\textit{OGImp}(#1,#2)}
\newcommand{\GI}{\mathit{GI}}

\newcommand{\gimp}[2]{\GI_{#1}(#2)}
\newcommand{\EC}{\mathit{EC}}

\newcommand{\Proj}[2][]{\textit{Inj}_{#1}(#2)}

\newcommand{\GL}[1]{\textit{GL}(#1)}

\newcommand{\PC}[2]{\textit{LP}(#1,#2)}
\newcommand{\PCB}[3]{\textit{LP}_{#1}(#2,#3 )}

\newcommand{\discomp}{\B_1 \bigtriangledown \B_2}

\newcommand{\discompthree}{\B_3 \bigtriangledown (\B_1 \bigtriangledown \B_2)}

\newcommand{\discomptwoone}{(\B_3 \bigtriangledown \B_1) \bigtriangledown \B_2}

\newcommand{\prm}[2]{ {#1}{'}}

\newcommand{\segarrow}{\leadsto }
\newcommand{\segtrans}[1]{\overset{#1}{\leadsto} }

\newcommand{\goaltrace}{{\sigma}}

\newcommand{\mapcv}{\rho}

\newcommand{\sattop}{\ell_{\top}}
\newcommand{\satbot}{\ell_{\bot}}

\newcommand{\etapi}{\eta^\ini_{\pi}}

\addto\extrasenglish{

}

\usepackage{multirow}

\jotdetails{
    volume=vv,      
    number=nn,       
    articleno=aa,   
    year=yyyy,      
    license=ccbyncnd    
}

\articletype{regular} 

\title{Disjunction Composition of BDD Transition Systems for Model-Based Testing}

 \author{Tannaz Zameni}
 \author{Petra van den Bos}
 \author{Arend Rensink}
 \affil{University of Twente, Enschede, The Netherlands}

\begin{abstract}
We introduce a compositional approach to model-based test generation in Behavior-Driven Development (BDD). BDD is an agile methodology in which system behavior is specified through textual scenarios that, in our approach, are translated into transition systems used for model-based testing.
This paper formally defines \emph{disjunction composition}, to combine BDD transition systems that represent alternative system behaviors.  
Disjunction composition allows for modeling and testing the integrated behavior while ensuring that the testing power of the original set of scenarios is preserved. This is proved using
a symbolic semantics for BDD transition systems, with the property that the symbolic equivalence of two BDD transition systems guarantees that they fail the same test cases. 
Also, we demonstrate the potential of disjunction composition by applying the composition in an industrial case study. 
\end{abstract}


\runningtitle{Disjunction Composition of BDD Transition Systems for Model-Based Testing}
\runningauthor{Zameni et al.}

\acknowledgment{This publication is part of the project \textit{TiCToC -Testing in Times of Continuous Change-} with project number 17936 of the research program \textit{MasCot-Mastering Complexity-} 
which is supported by NWO. 

This research was supported by NWO project OCENW.M.23.155 \emph{Evidence-Driven Black-Box Checking (EVI)}.}

\begin{document}

\maketitle
\urlstyle{rm}

\section{Introduction}
Testing often is the primary method for verifying system behavior and is a critical part of software development. But as systems grow in complexity, ensuring the correctness of all integrated behaviors becomes increasingly challenging. Model-based testing (MBT) offers a powerful solution by automating the generation of diverse test cases derived from formal models: precise specifications that define what the system should and should not do. Additionally, using partial models and composing them offers an effective way to ease model construction and obtain larger models that detect bugs that are difficult and expensive to find through manual testing. However, creating and maintaining such formal models can be time-consuming and requires specialized skills, which often makes MBT hard to adopt in everyday software development.


On the other hand, Behavior-driven Development (BDD) is a widely adopted agile methodology that offers a more accessible way to describe system behavior for testing. In BDD, different stakeholders, such as product owners, developers and testers, collaboratively discover system behaviors that are captured in structured documents. A common structure is the \emph{Given-When-Then} format of the Gherkin language \cite{bddbook-formulation,RSpec-book}: \emph{Given} a precondition for the required system state, \emph{When} an action is performed on/by the system, \emph{Then} specified actions and resulting state are expected. Tools like Cucumber \cite{Cucumber} and Reqnroll \cite{Reqnroll} enable users to manually implement the steps of BDD scenarios. However, they lack formal semantics: due to use of natural language, BDD scenarios may be ambiguous. Moreover, BDD scenarios lack mechanisms to automatically compose scenarios and capture their integrated behavior.

In this paper we combine BDD and MBT. We start from BDD Transition Systems (BDDTSs): formal models that can be derived from BDD scenarios (see, e.g., \cite{ZameniIntLang,ZameniBDDtoTestGen}). This paper introduces \emph{disjunction composition} for BDDTS. It combines BDDTSs that represent alternative valid ways for a system to behave. When applied, it integrates the overlapping, but partly different behaviors of two BDDTSs into one BDDTS.  In particular, if both BDDTS describe the same action in their \step{When} step, the composition takes this action as well. Actions that only appear in the \step{When} step of one of the BDDTS are included in the composition as well. This way we faithfully model that \step{When} steps describe which actions a system \emph{may} perform. In contrast, \step{Then} steps describe \emph{expected} behavior. The disjunction composition allows all actions of \step{Then} steps specified by both BDDTS. Other actions are discarded, since the BDDTS disagree on those expectations (and their tests would as well), hence they cannot both be true. (The consequence of discarding an action from the model is that, if a System under Test does perform that action, the test fails.) Hence, disjunction composition not only facilitates the integration of two BDDTS; it also reveals modeling inconsistencies.

Concrete advantages of generating tests on the basis of a composed BDDTS, which covers more behavior, are that the System under Test has to be initialised (i.e., set up) less frequently, and the resulting tests more often give a definite verdict (pass or fail) rather than being inconclusive (as happens when the system displays behavior not covered by the model).

The name \emph{disjunction composition} reflects the intuition that any behavior that either of the composed features can perform is allowed in the When steps, reflecting alternative (disjunctive) possibilities. However, some aspects of the construction (such as discarding non-overlapping actions of \step{Then} steps) actually have a more conjunction-like flavour.

\begin{figure}[b!]
    \centering
    \includegraphics[scale=0.5]{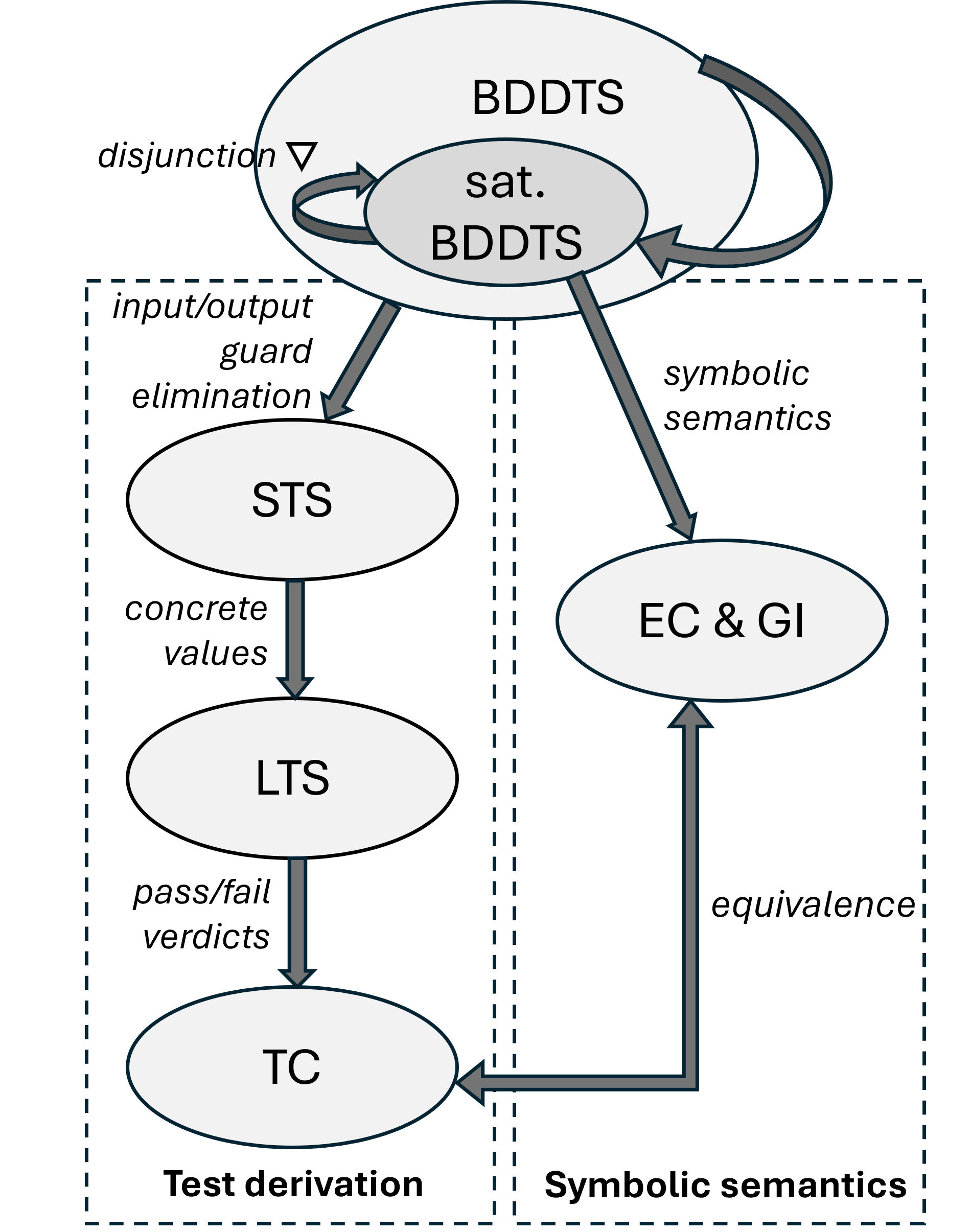}
    \caption{Overview of our paper}
    \label{fig:overall}
\end{figure}

This paper makes four main contributions: 
\begin{enumerate}[noitemsep]
    \item We formally define disjunction composition;
    \item We define a symbolic semantics for BDDTS to reason symbolically about equivalence for testing, and use this to show that the composition of two BDDTS is equivalent to the conjunction of the semantics of those two BDDTS;
    \item We prove correspondence between symbolic semantics of a BDDTS and the concrete test cases that can be derived from it;
    \item We apply disjunction composition on an industrial case study, where we apply and compare our approach with the Dutch Railway's approach of using BDD~scenarios for testing their information boards that display train departures.
\end{enumerate}
%
\autoref{fig:overall} shows an overview of the main elements of our paper.
At the top it shows BDDTS (\autoref{sec:transitionsystems}), the model of this paper.
The \emph{concrete semantics} of BDDTS (\autoref{sec:testcases}) on the left of the overview, is based on an interpretation from Symbolic Transition Systems (STSs) \cite{covsts} in terms of Labeled Transition Systems (LTSs) \cite{mbtlts}, where transitions are instantiated with concrete data values, similar to \cite{testgenerationSTS}. This semantics is well-suited for test case (TC) generation and understanding system behaviors at the instance level.
The symbolic semantics (\autoref{sec:sym-semantics}) describes the system in two ways. First, Execution Conditions (ECs) capture the conditions under which a sequence of transitions can happen. Second, Goal Implications (GIs) express how these conditions lead to expected outcomes, corresponding to the goals defined in the \step{Then} steps of a BDDTS. In other words, ECs tell us when things can happen, and GIs tell us what should be the result if they do happen. Symbolic semantics is based on saturated BDDTS (\autoref{sec:transitionsystems}), which is a version of the BDD~transition system where every possible output and the relevant inputs are explicitly represented. 
This semantics abstracts away from specific data values, allowing reasoning over sets of behaviors in a compact, generalized form. We use this to reason about equivalence of two BDDTS with their disjunction composition (\autoref{sec:disjunction}). Finally, we discuss the case study (\autoref{sec:casestudy-result}), related work (\autoref{sec:relatedwork}) and conclusions (\autoref{sec:conclusion}). Proofs of the results can be found in the appendices.
\section{BDD Transition Systems}
\label{sec:transitionsystems}
To define BDDTS formally, we need to introduce several notations, namely for functions, for terms and assignments, and for states and labeled transitions.

\paragraph{Functions.}
We write a function $f$ from domain $X$ to codomain $Y$ as $f:X \rightarrow Y$, and a partial function as $f: X \hookrightarrow Y$. We also write $Y^X$ for the set of all such functions, and $Y^X_\bot$ for the partial functions. With $f\definedon x$ we denote that $f$ is defined for $x$, and with $f \undefinedon x$ that $f$ is undefined for $x$. This is extended to $f\definedon Z$ and $f\undefinedon Z$ (with $Z\subseteq X$) to mean (un)definedness for \emph{all} $x\in Z$.

Given two partial functions $f_A: A \hookrightarrow B$ and $f_C: C \hookrightarrow D$ such that $f_A \downarrow x $ and $f_C \downarrow x$ implies $f_A(x)=f_C(x)$ for all $x \in  A \cap C$, we define their union $f_A \sqcup f_C: (A \cup C) \hookrightarrow (B \cup D)$ as:
\[f_A \sqcup f_C: x \mapsto
\begin{cases}
    f_A(x) & \text{if } x \in A \text{ and } f_A \downarrow x\\
    f_C(x) & \text{if } x \in C \text{ and } f_C \downarrow x.
\end{cases}
\]
%
%
\paragraph{Terms and assignments.}
In the following, we briefly review common algebraic concepts like 
terms, variables and assignments, and their semantics, without going into full formal detail.



\begin{itemize}
\item We build \emph{terms} from functions and variables (as usual). 
The set of all variables is denoted $\VU$, and the set of terms over variables $V\subseteq\VU$ is denoted 
$\TU(V)$.
A term $t$ is called \emph{ground} if $t\in \TU(\emptyset)$. 
We assume that terms are well-typed. Terms of type $\Bool$, over variables $V$, are denoted $\TU_\Bool(V)$.

\item The set of variables actually occurring in a term $ t \in  \TU(V)$ are denoted $\var(t)$.
\item  We write $t[x/y]$ with $t\in \TU(X)$ and $x,y\in X$ for the substitution of all $y$ by $x$ in $t$.
\item An \emph{assignment} is a partial function $a \in \TU(V)^Y_\bot$ for some $V, Y\subseteq \VU$. $\id_V \in \TU(V)^V$ is the \textit{identity assignment} over $V$, defined as $\id(v)=v$ for all $v \in V$. 
%
%
We will also use $a_1\sqcup a_2$ if $a_1$ and $a_2$ are \emph{compatible} ($a_1\compat a_2$, see below). We apply an assignment $a$ to a term $t$ by writing $t[a]$ (which syntactically substitutes, in $t$, all $x$ for which $a\definedon x$ by $a(x)$). This notation is extended pointwise to assignments: $a_1[a_2]$ is the assignment defined by $\{ x\mapsto a_1(x)[a_2] \mid a_1\definedon x\}$.

\item The value of a term $t\in \TU(V)$ is generated by a \emph{valuation} $\vu:V\rightarrow \UU$, where $\UU$ denotes the domain of (semantic) values (of all types).
%
For $t\in \TU(V)$ and $\vu\in \UU^V_\bot$ with $\vu\definedon\var(t)$ we write $\bracket t_\vu\in \UU$ to denote the value of $t$ under $\vu$, i.e. the value that is obtained by substituting the values for variables according to $\vartheta$ in the term $t$, and then evaluating the result.
Likewise, for assignments $a\in \TU(V)^V_\bot$, $\bracket a_\vu\in \UU^V_\bot$ is the valuation that maps each $x\in V$ to $\bracket{a(x)}_\vu$. If $t$ resp.\ the images of $a$ are ground, we don't need $\vu$ and will just write $\bracket t$ and $\bracket a$.

\item Terms $t_i\in \TU(V_i)$ for $i=1,2$ are \emph{semantically equivalent}, denoted $t_1\equiv t_2$, if $\bracket{t_1}_\vu=\bracket{t_2}_\vu$ for all $\vu\in \UU^{V_1\cup V_2}$. Assignments $a_i\in \TU(V_i)^{Y_i}$ for $i=1,2$ are \emph{semantically compatible}, denoted $a_1\compat a_2$, if $a_1(x)\equiv a_2(x)$ for all $x\in Y_1\cap Y_2$. If $a_1\compat a_2$ then $a_1\sqcup a_2$ is well-defined, as announced above.

\item For predicates $t_i\in\TU_\Bool(V)$ and valuations $\vu$ with $\vu\definedon\var(t_i)$ for $i=1,2$, we use $\vartheta\models t_1$ (``$\vartheta$ satisfies $t_1$'') as an alternative notation for $\sem{t_1}_{\vartheta}=\True$, and $t_1\Rrightarrow t_2$ to denote semantic implication, i.e., $\vu\models t_1$ implies $\vu\models t_2$ for all $\vu\in \UU^V$.
\end{itemize}
%
%
In the above, we ignore the choice of the particular functions and predicates from which terms are built. This is an issue treated in more detail in \cite{ZameniIntLang}, where we introduced the concept of a \emph{Domain-Specific Interpretation}: essentially an algebra that captures relevant information about the domain, such as the available types, values and operations on them.

\paragraph{States and labeled transitions.}
In this paper, we often use \emph{transition relations}, which are always defined as $\arrow\subseteq Q\times A\times Q$ for a set $Q$ of states (sometimes called locations) and a set $A$ of labels. For such transition relations, we adopt some common notations. To start, $q_0\trans{a_0\cdots a_{n-1}}q_n$ for $n\geq0$ denotes $(q_i,a_i,q_{i+1})\in{\rightarrow}$ for $0\leq i <n$; next, $q\trans\sigma P$ with $P\subseteq Q$ and $\sigma\in A^*$ denotes that there is some $q'\in P$ such that $q\trans\sigma q'$; and finally, $q\trans\sigma{}$ abbreviates $q\trans\sigma Q$. A state $q$ is called a \emph{sink}, denoted $\sink(q)$, if $\nexists a:q\trans a$. The subset of sink states in $P\subseteq Q$ is denoted $P_\sink$. 

A set of labels $A$ is usually partitioned in a set of inputs $A_i$  and outputs $A_o$, i.e. $A=A_i\cup A_o$, and $A_i \cap A_o = \emptyset$. An output denotes an action of the System under Test, that is provided to, or observed by its environment (e.g., a system user, another system, or the testing tooling). Conversely, an input denotes an action performed by the environment that is provided to the System under Test. 

In several cases, the set of states $Q$ will be partitioned into $Q=Q^\Op\cup Q^\Cl$ of \emph{open} and \emph{closed} elements, to record the intended interpretation of ``missing outputs'', i.e., output actions for which there is no outgoing transition from a given state. If a system actually performs such a missing output, contrary to what the model specifies, how should this be dealt with? From closed$^\Cl$ states this is considered forbidden (and will lead to a failing test) whereas from open$^\Op$ states this is considered unspecified (and will lead to inconclusive). These different interpretations are inspired by the meaning of actions specified in \step{then} vs.\ \step{when} steps of a BDD scenario. We call sets $Q$ that are partitioned in this way \emph{\OC-natured}. For a \OC-natured set $Q$ we will use $\N:Q\to\setof{\Op,\Cl}$ to retrieve the nature.

For the case of BDDTSs, defined below, the transitions are defined between \emph{locations} (which we denote $L$ rather than $Q$ as above), and its labels (which we denote $\Lambda$ rather than $A$) are structures consisting of several elements, among which are so called \emph{interactions}. The transitions themselves are called \emph{switches} to reflect their more complex structure. An interaction $\alpha = g\,\ivbar$ consists of
\emph{(i)} a \emph{gate} $g\in\GU$, where $\GU$ is a universe of gates, partitioned into input gates $\GU_i$ and output gates $\GU_o$, and
\emph{(ii)} a sequence of \emph{interaction variables} $\ivbar=\iv_{1}\cdots \iv_{n}\in \IV^*$, where $\IV\subseteq \VU$ is a special set of variables.
We use $g_\alpha$ and $\ivbar_\alpha$ to denote the constituent elements of an interaction $\alpha$. The universe of interactions $\IAU\subseteq \GU\times \IV^*$ forms a one-to-one relation, i.e., gates have fixed interaction variables. For a subset $G \subseteq \GU$, we use $\IAU(G)=\{\alpha\in \IAU\mid g_\alpha\in G\}$ to denote the subset of interactions with gates in $G$. In terms of the general setup discussed above, the partitioning of $\Lambda$ into input labels $\Lambda_i$ and output labels $\Lambda_o$ is determined by the input/output-nature of the interaction gate contained in the label. 

\paragraph{BDDTS.}
From these basic elements we now define the core notion of a \emph{BDD Transition System} (BDDTS), inspired by \cite{ZameniIntLang}, which formally captures the behavior prescribed by a single BDD scenario or (alternatively) the composition of several BDD scenarios. In \autoref{exmp:bddts} we will discuss an example of \autoref{def:BDDTS}.
\begin{definition}[BDD Transition System]\label{def:BDDTS}
A \emph{BDD Transition System} (BDDTS) is a tuple $\B = \tupof{L,V,G,\rightarrow,\il,\IG,\OG}$, where
\begin{itemize}[topsep=\smallskipamount]
\item $L$ is a \OC-natured set of \emph{locations};
\item $V$ is a set of \emph{location variables}, disjoint from $\IV$ and partitioned into $\LV$ (model variables) and $\CV$ (context variables);
\item $G\subseteq \GU$ is a set of \emph{gates}. We write $G_i=G\cap \GU_i$ and $G_o=G\cap \GU_o$, and use $\IV$ to denote the set of interaction variables of $G$ (and hence of $\B$).
\item $\arrow \subseteq L \times \Lambda \times L$ is a \emph{switch relation}, with label set $\Lambda=\IAU(G) \times \TU_\Bool(V\cup\IV) \times \TU(V\cup\IV)_\bot^{\LV}$;
\item $il\in L^\Op$ is the \emph{initial location} (which is open);
\item $IG \in \TU_{\Bool}(V)$ is the \emph{input guard} (which can also be thought of an \emph{initial guard});
\item $OG: L \hookrightarrow \TU_{\Bool}(V)$ is a partial map from locations to \emph{output guards}.
\end{itemize}
$\B$ is \emph{deterministic}, meaning that: if $t_i=(\sl,\alpha,\phi_i,a_i,\tl_i)\in \arrow$ for $i=1,2$, then either $t_1=t_2$ or $\phi_1\wedge \phi_2\equiv\False$. The \emph{active variables} $V_l\subseteq V$ of a location $l\in L$ are recursively defined as the smallest set such that $\var(\OG(l))\subseteq V_l$ if $\OG\definedon l$, and $\var(\phi_t)\cup V_{\tl_t}\subseteq V_{\sl_t}\cup \IV$ for all $t\in\arrow$. For all $t\in\arrow$ we have $a_t\definedon V_{\tl_t}$.
\end{definition}
We use $L_\B$, $V_\B$ etc.\ to denote the components of a BDDTS $\B$; for $\B_i$, we further abbreviate $L_{\B_i}$, $V_{\B_i}$ etc.\ to $L_i$, $V_i$. When considering multiple BDDTSs $\B_i$, we assume they are \emph{compatible}, i.e., share the same $G_i$ and $V_i$. Hence, we write simply $G$ and $V$. The components of a switch $t\in{\rightarrow}$ are denoted $\sl_t$ (the \emph{source location}), $\alpha_t=(g_t,\ivbar_t)$ (the \emph{interaction}), $\phi_t$ (the \emph{switch guard}, a boolean term), $a_t$ (the \emph{assignment}, a partial function from $\LV$ to terms over the active variables of the target location) and $\tl_t$ (the \emph{target location}); for a switch denoted $t_i$, we write $\sl_i$, $\alpha_i$ etc. We call a location $l$ a \emph{goal location} if $\OG \definedon l$.

An important aspect is the distinction between model variables and context variables. Context variables are assumed to be directly observable or read from the System Under Test (SuT); they reflect the actual runtime state or outputs of the system. Model variables, in contrast, belong to the formal model and represent expected or hypothesized state. As a consequence, model variables can be changed through assignments in BDDTS switches (such an assignment reflects an expectation of how the system state changes), in contrast to context variables.

A BDDTS does contain an initialization of its model variables or context variables; however, the Input Guard typically constrains their expected initial value (possibly even to a particular constant). Upon instantiation of a BDDTS (in our setting, when it is transformed to an STS and from there to a Test Case), initial values have to be provided that satisfy the Input Guard. When a context variable occurs in a switch guard, its value is taken to be unchanged with respect to the initial state; when it occurs in an output guard, however, it is compared with the actual runtime state at that moment. (How this comparison is performed is fixed in the transformation to an STS.)

While model variables are not inherently tied to the SuT, they may be compared against context variables --- e.g., in guards or assertions --- to express that the observed system behavior conforms to the model’s expectations.
\begin{figure*}[t!]
\begin{subfigure}[h]{0.35\textwidth}
\begin{scenarios}
\item \label{scenario:door} Door access with badge
    \begin{steps}
        \Given the badge ID 1234 is authorized,
        \And badge ID 1234 is presented at the door,
        \When the system verifies the badge,
        \Then the system sends the command to open the door
        \And the state of the door is now changed to open.
    \end{steps}
\end{scenarios}
\caption{\ref{scenario:door} describes the verification of a badge to get access through a door.}
\label{fig:door-scenario}
\end{subfigure}
\begin{subfigure}[h]{0.65\textwidth}
    \includegraphics[width=\linewidth]{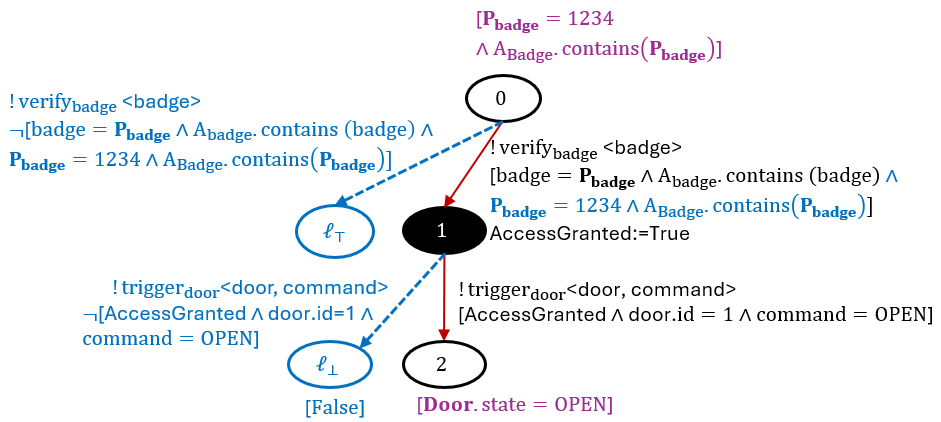}
    \caption{BDDTS $\B$ has open locations 0 and 2, with input and output guard resp., closed location 1, and the two red switches. The saturation of $\B$ also includes the dashed switches and blue expressions.}
    \label{fig:door-bddts}
\end{subfigure}
\caption{A BDD scenario (a) and its corresponding BDDTS (b)}
\end{figure*}
\begin{example} \label{exmp:bddts}
\autoref{fig:door-bddts} shows the BDDTS for \ref{scenario:door} of \autoref{fig:door-scenario}.
The \step{Given} step is represented by the input guard on location 0.
The When step is represented by the $\code{!verify_{badge}}$ output switch, and the first part of \step{Then}, i.e. before \step{And}, is represented by the $\code{!trigger_{door}}$ output switch. Observe that such a \step{Then} switch  starts from a closed location, while the switch resulting from a \step{When} step starts from an open location. The second part of the Then step, i.e. after \step{And}, is formalized as the output guard of location 2. 

Model variable $\code{A_{badge}}$ contains the list of all authorized badges, while $\code{P_{badge}}$ is a context variable of the badge provided in the \step{Given} step. $\code{P_{badge}}$ is a context variable as it must be read from the system, while $\code{A_{badge}}$ is a list stored in the model.
With model variable \code{AccessGranted} we memorize that access was verified in the $\code{!verify_{badge}}$ switch, so that this can be used to guard the $\code{!trigger_{door}}$ switch. The \code{Door} context variable provides the final status of the door, OPEN or CLOSED. 

\end{example}
BDDTSs were originally introduced as a basis for deriving test cases. In such a test case, the interpretation of a concrete system action that fails to satisfy any of the switch guards in a given location $l$ (in other words, a system action that is ``absent'' from the BDDTS) depends on the \OC-nature of $l$: for open states and for input actions, the test is inconclusive, whereas for output actions from closed states, the test fails. Instead of directly going from BDDTSs to test cases, however, we can first compose several (similar) BDDTSs, resulting in a more comprehensive BDDTS that combines the testing power of the original ones. (This is the core contribution of this paper.) For the composition to work as required, we need to first \emph{saturate} the BDDTS by adding switches to capture the intended interpretation of ``absent'' system actions.

\begin{definition}[saturated BDDTS] \label{def:saturated}
A BDDTS $\B$ is called \emph{saturated} if:
\begin{enumerate}[nosep]
\item\label{open-sat} For any $l\in L_\B$ and $\alpha \in \IAU(G)$, the set $\set{t_1,\ldots,t_n}\subseteq \arrow_\B$ of all switches with $\sl_i=l$ and $\alpha_i=\alpha$ is either empty or satisfies $\bigvee \phi_i \equiv \True$;

\item\label{closed-sat} For any $l\in L_\B^\Cl$ and $\alpha \in \IAU(G_o)$, there is a switch $l\trans{\alpha,\phi,a}_\B$ for some guard $\phi$ and assignment $a$.

\item\label{IG-sat} For any $\il_\B\trans{\alpha,\phi,a}_\B l$, either $\phi\Rrightarrow \IG$ or $l\in L_{\B,\sink}^\Op$ and $\OG\undefinedon l$.
\end{enumerate}
\end{definition}
Clause \ref{open-sat} requires that, for every location and interaction $\alpha$, either there is no outgoing $\alpha$-switch at all or there is a switch for all $\alpha$-based interaction values. Clause~\ref{closed-sat} requires that, for every \emph{closed} location and \emph{output} interaction $\alpha$, there is an outgoing $\alpha$-switch (so the first case of Clause~\ref{open-sat} does not occur). Clause~\ref{IG-sat}, finally, encapsulates the intuition that initial switches that correspond to behavior disallowed by the input guard always give rise to an inconclusive verdict. 
This is achieved by insisting that such switches immediately end in an open, non-goal sink location.

Turning an arbitrary BDDTS into a saturated one (i.e., \emph{saturating} it) consists of three steps:
\begin{enumerate}[noitemsep]
\item \textbf{Input guard propagation}: Since a BDDTS only applies to systems that satisfy its input guard, we make this explicit by adding the input guard to the guards of all existing initial switches. 
    
\item \textbf{Completion for existing interactions}: For all sets of switches $T$ from a certain location $l$ with the same interaction $\alpha$, we add a \emph{saturating switch} from $l$ for $\alpha$, with a guard equivalent to the negation of all $T$-guards.
    Consequently, after saturation, location $l$ has a switch whose guard evaluates to true for any valuation of variables.
Saturating switches from open locations lead to a state where any action is allowed ($\ell_\top$ in the definition below), whereas those from closed locations lead to a state ($\ell_\bot$) which has output guard $\False$ and hence is essentially a failure state.

\item \textbf{Completion for absent output interactions}: If $l$ in the above is a closed location, we also add switches for output gates that are not enabled, again to $\ell_\bot$, because these switches represent disallowed behavior.
\end{enumerate}
\begin{definition}[saturating a BDDTS] \label{def:saturation}
Given a BDDTS $\B = \tupof{L,V,G,\rightarrow,\il,\IG,\OG}$, its saturation is defined as 
\[ \sat{\B}= \tupof{L \cup \{\sattop,\satbot\},V,G,\rightarrow_\satK,\il,\IG,\sat{\OG}} \]
Here 
$\sattop,\satbot$ are fresh open locations such that $\sat\OG=\OG\sqcup \{(\ell_\bot,\False)\}$ (hence $\sat\OG\undefinedon \sattop$).
Switch relation $\arrow_\satK$ is constructed from $\IG$-enriched switches $\rightarrow_{IG}$ and saturating switches  $\arrow_\top$, $\arrow_\bot$ to (respectively) $\sattop$ and $\satbot$, using all enabled guards $\Phi_{l, \alpha}$ from a location $l$ for interaction $\alpha$, as follows:
\begin{align*}
\rightarrow_\IG
 & = \{t \in \arrow \mid \sl_t \neq il \} \\
 & \quad\cup \{(il,\alpha_t,IG \land \phi_t,a_t,\tl_t) \mid t \in \arrow, sl_t = il\}\\ 
 \Phi_{l, \alpha} &= \{\phi_t \mid t \in \arrow_{IG},\sl_t=l,\alpha_t=\alpha\}: \\
\arrow_\top
 & = \{(l, \alpha,  \neg\textstyle\bigvee \Phi_{l,\alpha}, \epsilon, \sattop) \mid (l \in L^\Op \vee g_\alpha\in G_i), l\trans{\alpha,\phi,a} {}\}\\
\arrow_\bot
 & = \{(l, \alpha,  \neg\textstyle\bigvee \Phi_{l,\alpha}, \epsilon, \satbot) \mid l \in L^\Cl, g_\alpha\in G_o \}\\
\arrow_\satK
 & = \arrow_{IG} \cup \arrow_\top \cup \arrow_\bot
\end{align*} 
\end{definition}
Note that, as a special case, $\arrow_\bot$ contains $(l,\alpha,\True,\epsilon,\satbot)$ for all $l\in L^\Cl$ and $\alpha\in \IAU(G_o)$ such that $\nexists t\in\arrow: \sl_t=l\wedge \alpha_t=\alpha$. $\sat\B$ is deterministic as required by \autoref{def:BDDTS}, and the following proposition holds by construction.
\begin{proposition}
For any BDDTS $\B$, $\sat\B$ is saturated.
\end{proposition}
In \autoref{fig:door-bddts}, two dashed blue switches are added through saturation. Also, the input guard is conjoined, in blue, to the guard of the red switch from location 0.

\section{Semantics of BDDTS in terms of Goal Implications}
\label{sec:sym-semantics}
The essence of saturated BDDTSs is captured by the traces (i.e., sequences of interactions) leading to \textit{goal locations}. We capture the semantics of such traces by collecting the successive switch conditions into a so-called \emph{path condition} --- similar to the path conditions for STSs in \cite{covsts}. To define this formally, we need some more notation.

\begin{itemize}
\item To refer to the value of a variable $x$ at some point in the past, we will use $x\attime i$ with $i\in \mathbb{N}$. To be precise, we assume that for all $x$ and all $i>0$, $x\attime i$ is a distinct (fresh) variable referring to the value of $x$ at $i$ time instances in the past. $x\attime 0$ is the same as $x$. In addition, $x\upshift$ \emph{up-shifts} the time index, hence $x\attime i\upshift=x\attime{i+1}$.

\item Up-shifting is extended to expressions: $e\upshift$ denotes a copy of $e$ where all variables have been up-shifted;  in $e\upshift_X$ 
only variables $X$ have been up-shifted (hence the suffix $\upshift_X$ can be thought of as applying an assignment $\{ x^{(i)}\mapsto x^{(i+1)} \mid x\in X, i\in \mathbb{N} \}$).

\item We will use $\sigma\in \IAU(G)^*$ to denote sequences of interactions, and $\pi\in \Lambda^*$ to denote \emph{paths}, i.e., sequences of switch labels.
For a given $\pi=\lambda_1\cdots\lambda_n$, the corresponding sequence of interactions is denoted $\sigma_\pi=\alpha_1\cdots\alpha_n$.
\item A \emph{gate value} for $g\in \GU$ is a tuple $g\,\bar w$, where $\bar w=w_1\cdots w_n\in \UU^*$ is a sequence of values, such that the number of values and their types match with $g$'s interaction variables $\ivbar$ (from interaction $(g,\ivbar) \in \IAU(G)$). We write $\GUU$ for the set of all gate values. 
 \item Given a gate value sequence $ \omega=g_1\,\bar w_1, \ldots, g_n\,\bar w_n $, we write $\sigma_\omega$ to denote its corresponding interaction sequence such that $(g_i,\ivbar_i) \in \IAU(G)$ for $1\leq i \leq n$.
 \item Because $\ivbar$ is fixed for each $g$ (via its interaction $(g,\ivbar) \in \IAU(G)$), a gate value $g\,\bar w$ uniquely gives rise to a valuation $\vu_{g\,\bar w}$ mapping each $\iv_i$ to the corresponding $w_i$.
\end{itemize}
A path condition only contains interaction variables as free variables. Values of model variables are completely determined by their initial value in the initial location, plus assignments in subsequent switches. We construct a symbolic assignment function $a_\pi^\ini$ that maps model variables to expressions over interaction variables (on path $\pi$) and symbolic initial values as assigned by $\ini$.  
Values of context variables remain constant along the path, as fixed by the initial assignment $\ini$.  
We use upshifting to distinguish between interaction variables used on path $\pi$, since switches may use the same interaction variables.
\begin{definition}[Path condition] \label{def:path-condition}
Given a BDDTS, for arbitrary paths $\pi \in \Lambda^*$ and initial assignments $\ini \in \TU(\emptyset)^V$, the symbolic assignment $a_\pi^\ini\in \TU(\IV)^V$ and path condition $\eta_\pi^\ini \in \TU(\IV)_\Bool$ are defined inductively over the length of $\pi$:

\begin{minipage}{\columnwidth}
\begin{align*}
     a_\epsilon^\ini &= \ini & \eta_{\epsilon}^\ini &= \True \\
        a_{\pi\cdot(\alpha,\phi,a)}^\ini &= a[a_\pi^\ini\upshift] & \eta_{\pi\cdot(\alpha,\phi,a)}^\ini &= \eta_\pi^\ini\upshift \wedge \phi[a_\pi^\ini\upshift]
\end{align*}
\end{minipage}
\end{definition}
The intuition behind a path condition is that it is fulfilled precisely by those sequences of gate values that satisfy the successive switch conditions. 
For example, a path condition for the sequence of red switches in \autoref{fig:door-bddts} is:
\begin{align*}
& \code{badge}^1 = \ini(\code{P_{badge}}) \wedge \ini(\code{A_{badge}})\code{.contains(badge^1)}\\
& {}\wedge \code{P_{badge} = 1234} \wedge \ini(\code{A_{badge}).contains(\ini(P_{badge}))}\\
& {}\wedge \True \wedge \code{door.id} = 1 \wedge \code{command = OPEN}
\end{align*}
To formulate the symbolic semantics of a BDDTS, we collect goal paths: all paths $\pi$ to some goal location $gl$ sharing an interaction sequence $\sigma$. The goal implication then expresses that for each goal path $(\pi,gl)$ of interaction sequence $\sigma$, the path condition of $\pi$ should make the output guard of $gl$ true. Intuitively, this expresses that taking all switches of $\pi$ should make the output guard at the end true; in BDD terms, this means that the postcondition of the \step{Then} step is satisfied after executing the scenario.
%

\begin{definition}[$\EC$ and $\GI$]\label{def:goal-implication}
Let $\B$ be a saturated BDDTS and  $\ini \in \TU(\emptyset)^V$ an initialization. The \emph{execution condition function} $\EC_{\B,\ini}\colon \IAU(G)^* \rightarrow \TU_\Bool(V)$ and \emph{goal implication function} $\GI_{\B,\ini}\colon \IAU(G)^* \rightarrow \TU_\Bool(V)$ are defined, for any $\sigma \in \IAU(G)^*$, as:
\begin{align*}
\EC_{\B,\ini}: \sigma & \mapsto \IG \land
  \textstyle\bigvee \bigl\{\, \eta_\pi^\ini \;\big|\; \il\trans\pi,\; \sigma_\pi = \sigma \,\bigr\} \\
\GI_{\B,\ini}: \sigma & \mapsto
  \textstyle\bigwedge \bigl\{\, \eta_\pi^\ini\Rightarrow\OG(l)[a_\pi^\ini] \\ 
  & \qquad\qquad \;\big|\; \il\trans\pi l,\; \OG\downarrow l,\; \sigma_\pi = \sigma \,\bigr\} 
\end{align*}
\end{definition}
Hence, for every $\sigma$, the execution condition $\EC$ and goal implication $\GI$ are predicates over (upshifted) variables $\IV$. The first characterizes when $\sigma$ is executable, whereas the second expresses the property that, when $\sigma$ leads to a goal location $l$, the output guard of $l$ holds whenever $\sigma$ is executable (see \autoref{ex:2} and~\ref{ex:3}).

Note that we have defined \autoref{def:goal-implication} only for saturated BDDTSs. Though the execution conditions and goal implications themselves can be computed for arbitrary models, their intention is to capture the testing power of BDDTSs (\autoref{thm:sym-concrete-relation} and~\autoref{corr:testing-equivalence}); but that correspondence only holds for saturated models. Indeed, the saturation of a BDDTS in general preserves neither $\EC$ nor $\GI$.\todo{Does this need an example?}

Now we can use the execution condition and goal implication as symbolic BDDTS semantics. We define that on the level of \emph{sets} $\BB$ of BDDTSs. As a convenient notation to restrict such a set to those BDDTSs of which the input guard is satisfied by a given initialization $\ini\in \TU(\emptyset)^V$, we use:
\[   \BB\restr\ini = \{\B\in\BB\mid \sem\ini \models \IG_\B \} \enspace.\]

\begin{definition}[Testing equivalence]\label{def:testing-equivalence}
Two sets of saturated BDDTSs $\BB_1,\BB_2$ are \emph{testing equivalent}, denoted $\BB_1\simeq \BB_2$, if for all $\ini \in \TU(\emptyset)^V$ and $\sigma\in \IAU(G)^*$:

\begin{minipage}{\columnwidth}
\begin{align*}
\textstyle \bigvee_{\B\in\BB_1\restr\ini} \EC_{\B,\ini}(\sigma)
 & \equiv \textstyle \bigvee_{\B\in\BB_2\restr\ini} \EC_{\B,\ini}(\sigma) \\
\textstyle \bigwedge_{\B\in\BB_1\restr\ini} \GI_{\B,\ini}(\sigma)
 & \equiv \textstyle \bigwedge_{\B\in\BB_2\restr\ini} \GI_{\B,\ini}(\sigma) \enspace.
\end{align*}
\end{minipage}
\end{definition}
One of the main results of this paper (\autoref{corr:testing-equivalence}) is that 
testing equivalence guarantees that the derived tests always give the same test verdicts on the same behavior. In other words, sets of BDDTSs that are testing equivalent have the same testing power, i.e., the same capability to accept or reject Systems under Test on the basis of their observable behavior. 

\section {Disjunction Composition}
\label{sec:disjunction}
We now come to the main contribution of this paper: the definition of disjunction composition $\bigtriangledown$ for saturated BDDTSs.

\newcommand{\ERROR}{\mathrm{ERROR}}
\begin{definition} {\label{def:disjunction}}
Let $\B_1$ and $\B_2$ be two saturated BDDTSs Their disjunction composition is then defined as
 $\B_1 \bigtriangledown \B_2= \bddtstuple$, where
\begin{itemize}[noitemsep]
\item $L = (L_1 \cup \{\bot\}) \times  (L_2 \cup \{\bot\})$, with $L^\Op=(L_1^\Op\cup\{\bot\})\times (L_2^\Op\cup\{\bot\})$ 

\item $\arrow$ is the relation generated by the following rules:
\begin{enumerate}
\item Disjunction for shared interactions:
\[
	\dfrac{
		l_1 \trans{ \alpha, \phi_1  ,a_1 } l'_1 
        \quad l_2 \trans{\alpha, \phi_2,a_2 } l'_2 
        \quad a_1\compat a_2
	}
	{
        (l_1,l_2) \trans{ \alpha, ( \phi_1  \wedge \phi_2), a_1 \sqcup a_2 } (l'_1, l'_2)
	}
   \enspace  \textrm{1}
 \]

 \item Disjunction for non-shared interactions:
 \[
	\frac{
		l_1 \trans{ \alpha, \phi_1  ,a_1 } l'_1
        \quad l_2\ntrans\alpha
	}
	{
		(l_1,l_2) \trans{ \alpha, \phi_1,{a_1} } (l'_1,\bot) 
	}
  \enspace  \textrm{2-1}\]
  \[
	\frac{
		l_1\ntrans\alpha
		\quad l_2 \trans{\alpha, \phi_2 ,a_2 } l'_2 
  }
	{
		(l_1,l_2) \trans{\alpha, \phi_2  ,{a_2} } (\bot,l'_2) 
	}
  \enspace  \textrm{2-2}
	\]
\end{enumerate}

\item $il=(il_1,il_2)$
\item $IG= IG_1 \lor IG_2$
\item $\begin{array}[t]{@{}l}
     OG: (l_1, l_2)\mapsto \\
    \qquad\begin{cases}
		OG_1(l_1)
         & \text{if $OG_1  \downarrow l_1 $ and $ OG_2 \uparrow l_2$}
        
        \\
		OG_2(l_2)
        & 
        \text{if $OG_1  \uparrow l_1 $ and $ OG_2 \downarrow l_2 $}
        \\
		OG_1(l_1) \land OG_2(l_2) & 
        \text{{if $OG_1  \downarrow l_1 $ and $ OG_2 \downarrow l_2$}}
		
	\end{cases}\end{array}$
\end{itemize}
%
$\B_1\bigtriangledown B_2$ is considered to be \emph{ill-defined} if for any reachable $(l_1,l_2)\in L_1\times L_2$ there are $l_1 \trans{ \alpha, \phi_1  ,a_1 }$ and $l_2 \trans{\alpha, \phi_2,a_2 }$ such that $a_1\ncompat a_2$.
\end{definition}
In the sequel, we only consider disjunction compositions that are well-defined.

Some explanation is in order. The locations of a composed BDDTS $B_1 \bigtriangledown B_2$ are tuples $(l_1,l_2)$ of locations from $\B_1$ and $\B_2$, where either $l_i$ can be $\bot$ if the behaviour has moved outside $\B_i$. The switch relation for a location $(l_1,l_2)$ is constructed from the switches from $l_1$ (in $\B_1$) and $l_2$ (in $\B_2$), using two types of rules. Rule~1 applies if both locations have switches for the same interaction $\alpha$ and compatible assignments ($a_1\compat a_2$): such switches are synchronised while their guards are conjoined. Note that, due to saturation of the $\B_i$, if one such switch exists, then others are guaranteed to exist (for the same interaction) with complementary switch guards, causing the set of switches derived by Rule~1 to be saturated once more. Rules 2-1 and~2-2, on the other hand, deal with interactions for which either $l_1$ or $l_2$ has a switch but the other does not (which includes the case where one of the locations equals $\bot$): then the existing switch is copied from the relevant component.
As for the output guard of $(l_1,l_2)$, this is in principle the conjunction of the defined output guards of the $l_i$ according to $\B_i$.

If there is a reachable location $(l_1,l_2)$ with switches featuring the same interaction but incompatible assignments ($a_1\ncompat a_2$), this invalidates the composition. This means that $\bigtriangledown$ is only partially defined. All our results about disjunction composition are implicitly restricted to cases where it is defined.

\autoref{prop:satcomp} and \autoref{prop:disjunctioncommassoc} state expected, but essential, properties of disjunction composition.

\begin{restatable}{proposition}{saturatedcomposition} \label{prop:satcomp}
    Disjunction composition yields a saturated BDDTS.
\end{restatable}
\vspace*{-7pt}
\begin{restatable}{proposition}{disjunctioncommassoc} \label{prop:disjunctioncommassoc}
    Disjunction composition is commutative and associative.
\end{restatable}
\begin{figure*}[t!]
    \centering
    \begin{subfigure}{0.48\textwidth}
\includegraphics[scale=0.55]{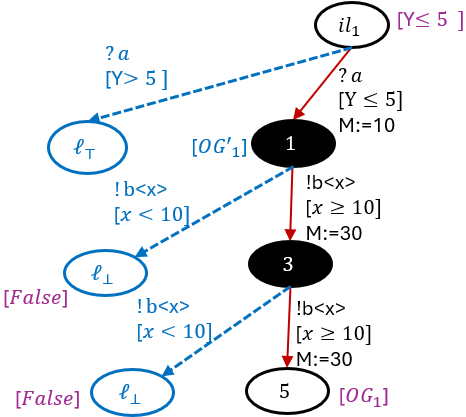}
        \caption{$B_1$}
        \label{fig:B1}
    \end{subfigure}
    \begin{subfigure}{0.48\textwidth}
        \includegraphics[scale=0.55]{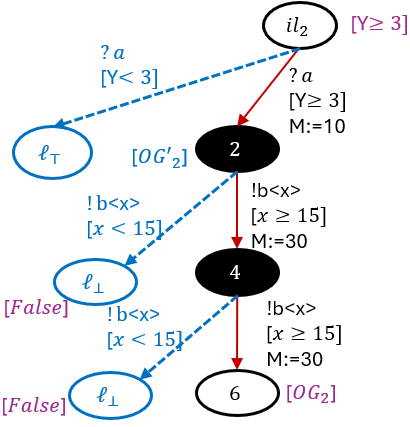}
        \caption{$B_2$}
        \label{fig:B2}
    \end{subfigure}
    \begin{subfigure}{\textwidth}
        \centering
\includegraphics [scale=0.55]{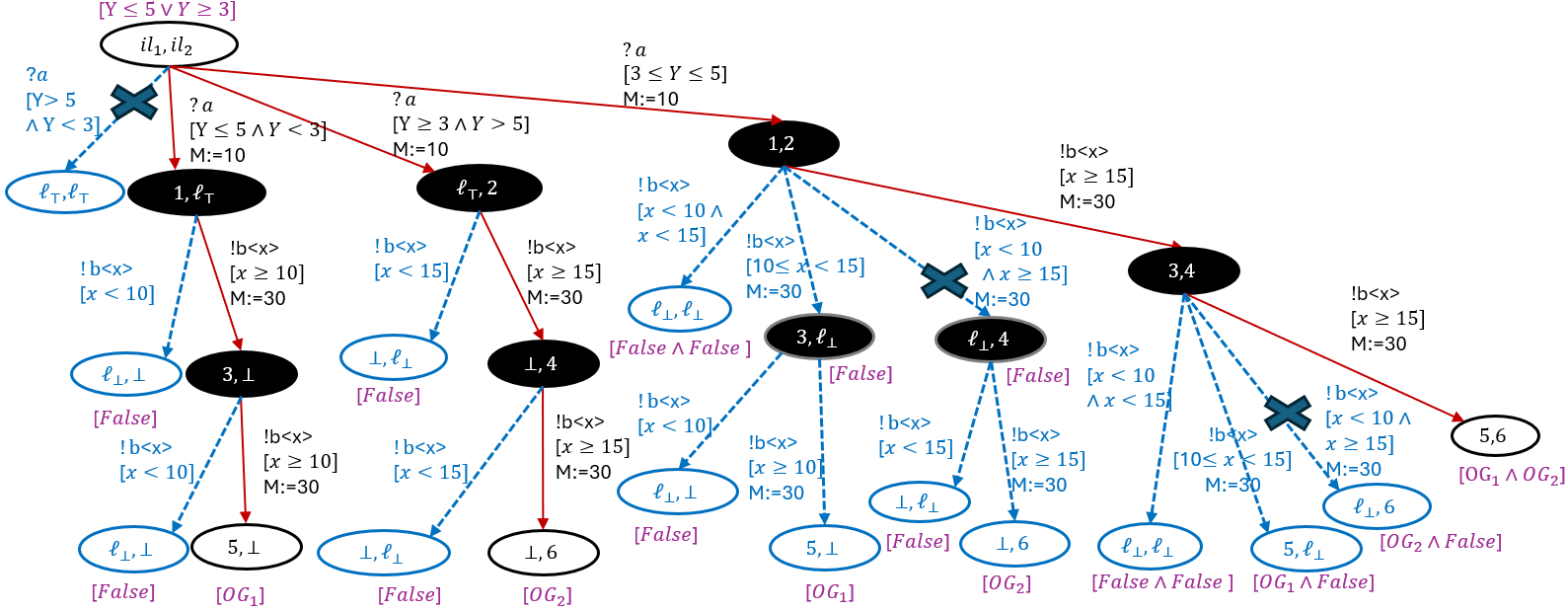}
\caption{$ B_1 \bigtriangledown B_2 $}
\label{fig:B1orB2}
    \end{subfigure}
    \caption{Two BDDTS and their disjunction composition}
    \label{fig:bddts-overlapping-guards}
\end{figure*}
\vspace*{-7pt}
\begin{example}\label{ex:2}
\autoref{fig:bddts-overlapping-guards} shows the disjunction composition of two BDDTS $B_1$ and $B_2$. 
As an example we also build the $\EC$ and $\GI$ for $\sigma=?a!b!b$ in $B_1$, 
the $\EC$ and $\GI$ for $B_2$ can be constructed similarly: 
\small{
\begin{align*}
\EC_{B_1,\ini}(\sigma) & = (\ini(Y)\leq5 \land x^{(1)}\geq 10 \land x \geq 10) \\
 & {} \lor (\ini(Y)\leq5 \land x^{(1)}\geq 10 \land x < 10),\\
\GI_{B_1,\ini}(\sigma) & = ((\ini(Y)\leq5 \land x^{(1)}\geq 10 \land x \geq 10)\Rightarrow\OG_1)  \\
 & {} \land ((\ini(Y)\leq5 \land x^{(1)}\geq 10 \land x < 10) \Rightarrow \False )
\end{align*}}
\end{example}
We can now state our first main theoretical result: the composition of two BDDTSs is testing equivalent to the two original BDDTSs. This is the justification of the construction: by using the composed $\B_1 \bigtriangledown \B_2$ rather than the individual $\B_i$, we can derive a combined test with exactly the same testing power.

\begin{restatable}{theorem}{disjunctiontraceequive} \label{thm:disjunctiontraceequive}
    Let $\B_1 $ and $\B_2$ be two saturated BDDTSs. Then it holds that: \[\{\B_1,\B_2\}\simeq \{\B_1 \bigtriangledown \B_2\} \enspace. \]
\end{restatable}

\begin{example} \label{ex:3}
Below, we show a small part of $\EC$ and $\GI$ for composition of \autoref{fig:B1orB2}. Note that for $\sigma=?a!b!b$, the $\EC$ and $\GI$ of individual BDDTSs are equivalent to the composition as stated in \autoref{thm:disjunctiontraceequive}.
\small{
\begin{align*}
\EC_{{B_1\bigtriangledown B_2},\ini}(\sigma)
    & = (\ini(Y)<3 \land x^{(1)}\geq 10 \land x \geq 10) \\
    & {} \lor (\ini(Y)<3 \land x^{(1)}\geq 10 \land x < 10) \\
    & {} \lor \ldots \lor (3\leq\ini(Y)\leq5 \land x^{(1)}\geq 15 \land x \geq 15) \\
\GI_{{B_1\bigtriangledown B_2},\ini}(\sigma) 
    & = ((\ini(Y)<3 \land x^{(1)}\geq 10 \land x \geq 10) \Rightarrow \OG_1) \\
    & {} \land ((\ini(Y)<3 \land x^{(1)}\geq 10 \land x < 10) \Rightarrow \False) \\
    & {} \land \ldots \land ((3\leq\ini(Y)\leq5 \land x^{(1)}\geq 15 \land x \geq 15) \\
    & \qquad\qquad \Rightarrow \OG_1 \land \OG_2)
\end{align*}
}
\end{example}

\section{Concrete Semantics and Test Cases}
\label{sec:testcases}
In this section we show how to construct test cases from arbitrary BDDTSs (not just saturated ones). As outlined in \autoref{fig:overall}, this is a three-step process: first we translate a BDDTS to a so-called \emph{symbolic transition system} (STS) to encode the input and output guards; then we use the standard semantics for STSs to obtain a \emph{labelled transition system} (LTS); finally, the \OC-nature of locations is used to convert this into a \emph{test case}.

\begin{definition}[STS] \label{def:STS}
A \emph{Symbolic Transition System} (STS) is a tuple $\S = \tupof{L, V, G, {\rightarrow}, \il, \ini}$, where the first five components are the same as those of a BDDTS, with the proviso that $V = \LV$ (hence $\CV=\emptyset$), and $\ini \in \TU(\emptyset)^V$ represents the initial (syntactical) assignment of the location variables.
\end{definition}
For translating BDDTSs into STSs, several approaches were informally discussed in \cite{ZameniIntLang}, using different ways to obtain the values of the context variables when the output guards are checked. One of those approaches, which adds extra check-and-retrieve switches, was worked out in \cite{Zameni-sequential}. Here we present a different approach: we \emph{assume} that the domain provides output interactions with interaction variables that retrieve the up-to-date values of the context variables when they are needed. In particular, every switch leading up to a goal location is assumed to have such special interaction variables. This enables output guards to be evaluated as part of the incoming switch guards. In addition, we assume that every transition leading to a goal location originates in a closed location. We call BDDTSs that satisfy this assumption \emph{output-rich}.

\begin{definition}[Output-rich] A BDDTS $\B$ is \emph{output-rich} if
\begin{enumerate}
\item for all $(g,\iv_0\cdots\iv_n)\in \IAU(G_o)$, there is an injective partial renaming $\mapcv^g: CV\hookrightarrow \{\iv_0,\cdots,\iv_n\}$, and
\item for all $t\in \arrow$, if $\OG\definedon\tl_t$ then $\sl_t\in L^\Cl$ and $\mapcv^{g_t}\definedon x$ for all $x\in\CV\cap \var(\OG(\tl_t))$.
\end{enumerate}
\end{definition}
Though these assumptions make our BDDTS-to-STS conversion less general than the approach of \cite{Zameni-sequential}, the assumptions are in fact fulfilled in well-known architectures such as APIs with a request-response communication pattern, like our case study (\autoref{sec:casestudy-result}). In such domains, our approach is more concise.
\begin{definition} \label{def:bddtstosts}
Let $\B$ be an output-rich BDDTS with renaming $\mapcv$, and let $\ini \in \mathcal{T(\emptyset)}^V$ such that $\sem\ini \models \IG$. 
The STS of $\B$ for $\ini$ is given by:
\begin{align*}
\ISS{\B}{\ini} &= \tupof{L, V,G , \arrow_c, il, \ini} \text{ where:}\\
\arrow_c =\ &  \{(sl_t,\alpha_t, \phi_t, a_t, tl_t) \mid t \in \arrow, \OG\undefinedon \tl_t \} \\
 \cup\ & \{ (sl_t, \alpha_t, \phi_t \wedge \OG[\mapcv^{g_t}], a_t, tl_t) \mid  t \in  \arrow,\ \OG \definedon tl_t \}
\end{align*}
\end{definition}

\begin{example}
    \autoref{fig:door-iss} shows the STS for the saturated BDDTS of \autoref{fig:door-bddts}. 
    In this example,
    we asume that ini:
        $\code{P_{badge}:=1234}$,
        $\code{A_{badge}:=[1234,5678,0666]}$,
        $\code{AccessGranted:=False}$,
        $\code{Door(id:=1,state:=CLOSED)}$
        and $\mapcv\code{(Door)=door}$
There are two locations with output guards in its BDDTS \autoref{fig:door-bddts}: $\code{\OG(2)=}$ $\code{Door.state=OPEN}$ and $\code{\OG(\ell_\bot)=False}$. In the STS they are both moved and conjoined to the switch guard and in switch with gate $\code{!trigeer_{door}}$ the context variable \code{Door} is substituted with interaction variable \code{door}. 
\end{example}
Note that in our setting, there is no need to reassign former context variables once they become model variables in the STS, since their values are already obtained through interaction variables and the corresponding model variables are not used afterwards.
\begin{figure*}[t!]
\begin{subfigure}{0.48\textwidth}
    \centering
    \includegraphics[scale=0.55]{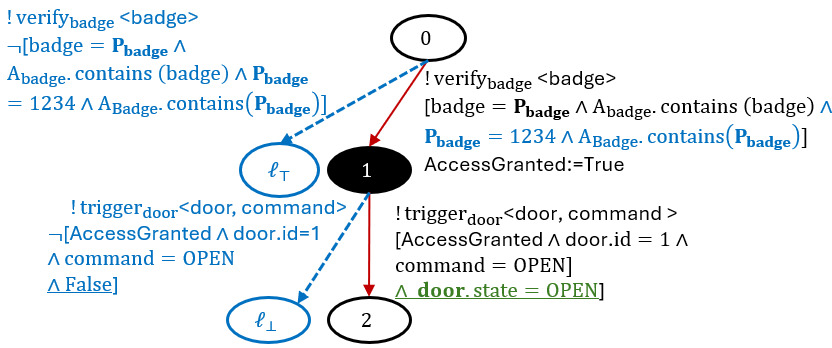}
    \caption{STS for BDDTS $\B$ of \autoref{fig:door-bddts}}
    \label{fig:door-iss}
\end{subfigure}
\qquad
    \begin{subfigure}{0.48\textwidth}
        \centering
    \includegraphics[width=\linewidth]{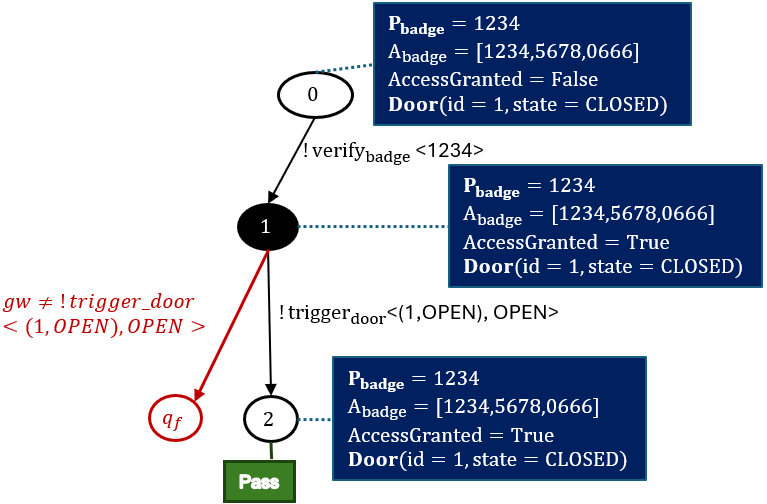}
    \caption{Test case for STS of \autoref{fig:door-iss}}
    \label{fig:door-test-case}
    \end{subfigure}
    \caption{Translation from BDDTS to STS to test case}
\end{figure*}
The next step is to translate (\OC-natured) STSs into (\OC-natured) LTSs, as in \cite{covsts}. First we define the general notion of LTS.

\begin{definition}[LTS]
A \emph{labelled transition system} (LTS) is a tuple $\L = \tupof{Q, A, \rightarrow, q_0}$, where:
\begin{itemize}
\item $Q$ is a set of states with designated initial state $q_0 \in Q$
\item $A$ is a set of labels, partitioned into input labels $A_i$ and output labels $A_o$
\item $\rightarrow\ \subseteq Q \times A \times Q$ is a transition relation.
\end{itemize}
\end{definition}
Next, we recall the STS semantics in terms of LTSs from \cite{covsts}, adding the \OC-nature.
%
\begin{definition}
\label{def:intrp}
Let $\S = \tupof{L,V,G,{\rightarrow},\il,\ini}$ be an \OC-natured STS. The \emph{semantics} of $\S$ is defined as the \OC-natured LTS  $\sem{\S} = \tupof{L\times \UU^V, \GUU,\arrow', (il,\bracket{\ini})}$ with $\N(l,\vu) = \N(l)$ for all $(l,\vu) \in L \times \UU^V$, and
\begin{align*}
\arrow' & = \{((\sl_t,\vu), u, (\tl_t, \sem{a_t \sqcup id_t}_{\vu_u\sqcup \vu})) \\
   & \qquad\mid 
   q \in L\times \UU^V, u\in \GUU, t \in \arrow, \vu_u\sqcup \vu \models \phi_t \}
\end{align*}
where $id_t(v) = v$ for all $\{v \in V \mid a_t \undefinedon v\}$ and $id_t \undefinedon v$ otherwise.
\end{definition}
In the final step, \OC-natured LTSs are converted into \emph{test cases}, defined as LTSs with designated \emph{pass} and \emph{fail} sink states.
\begin{definition}\label{def:testcase}
    A \emph{test case} is tuple $\tupof{Q,A,\arrow,q_0,P,F}$ where $\tupof{Q,A,\arrow,q_0}$ is an LTS and $P,F\subseteq Q_\sink$ are special sets of \emph{pass} and \emph{fail}  states.
\end{definition}
A test case $\TC$ is evaluated as follows: the behavior of a SuT is said to \emph{pass} $\TC$ if it leads \emph{through} a state in $P$ and to \emph{fail} it if it leads through a state in $F$. If a system neither passes nor fails (never reaching a state in $P\cup F$), the test is said to be \emph{inconclusive}. Formally:
\begin{align*}
\TC\passes \omega & \;\Leftrightarrow\; \exists \xi\preceq\omega: q_0\trans\xi P \\
\TC\fails \omega & \;\Leftrightarrow\; \exists \xi\preceq\omega: q_0\trans\xi F \enspace.
\end{align*}
To construct a test case from an \OC-natured LTS $\tupof{Q, A, \arrow, q_0}$, we introduce a fresh sink state $q_\f$, which is the only \emph{fail} state, to which we add new transitions from each closed state $q \in Q^\Cl$ for every action $a\in A_o$ that is not enabled from $q$. All open sink states in $Q$ are \emph{pass}.
As an example, \autoref{fig:door-test-case} shows the test case of BDDTS $\B$  from \autoref{fig:door-bddts}. Note that because BDDTSs are deterministic, so are the resulting STSs, LTSs and test cases.
\begin{definition} \label{def:gherkintestcase}
    Let $\L=\tupof{Q,A,\arrow,q_0}$ be an \OC-natured LTS and $q_\f\notin Q$. The test case for $\L$ is defined by $\TC(\L)=\tupof{Q\cup\setof{q_\f},A,\arrow',q_0,Q^\Op_\sink,\{q_\f\}}$ with
    \[ 
    \arrow' = \arrow 
       \cup \gensetof{(q,a,q_\f)}{q\in Q^\Cl,a\in A_o, q\ntrans a} 
    \]
\end{definition}
This chain of transformations from BDDTSs to test cases then yields:
\begin{definition}
The test case for an output-rich BDDTS $\B$ with initialisation $\ini \in \TU(\emptyset)^V$ such that $\sem\ini \models \IG$ is defined as:
\[ \TC(\B,\ini) \;=\; \TC(\sem{\ISS \B\ini}) \enspace. \]
\end{definition}
Note that, in contrast to the symbolic semantics and the disjunction composition, the test case construction does not require BDDTSs to be saturated. It is therefore important to note that saturating as in \autoref{def:saturation} does not change the fail verdicts of the derived test case.

\begin{restatable}{proposition}{propsaturated}\label{prop:saturation-testing}
If $\B$ is an output-rich BDDTS and $\ini \in \TU(\emptyset)^V$ such that $\sem\ini\models \IG$, then for any gate value sequence $\omega\in \GUU^*$
\[ \omega\fails\TC(\B,\ini) \enspace\Leftrightarrow\enspace \omega\fails\TC(\sat\B,\ini) \enspace. \]
\end{restatable}
\noindent
Unfortunately, \autoref{prop:saturation-testing} does not extend to pass verdicts: it turns out that saturation may in fact change the pass verdicts of the derived tests, either by turning a pass verdict into inconclusive or the other way around. This is a shortcoming that we plan to address in the future.

\medskip\noindent
Our second main result, \autoref{thm:sym-concrete-relation}, relates the symbolic semantics of BDDTSs (\autoref{def:goal-implication}) to its concrete semantics in terms of test cases (\autoref{def:gherkintestcase}).

To state this formally, we need some additional notation. For a valuation $\vu\in \UU^\IV_\bot$, with $\upshift\vu\in \UU^{\IV\upshift}_\bot$ we denote a variant that applies $\vu$ to an upshifted domain. This is formally defined by:
\[ \upshift\vu: x\mapsto \vu(y) \quad \text{if } x=y\upshift \enspace. \]
(Note the difference with $a\upshift$, in which the upshift is applied to the \emph{result} of $a$.) A gate value sequence $\omega$ then gives rise to a valuation $\vu_\omega$ for interaction variables that are increasingly upshifted for ``older'' interactions, inductively defined by
\[ \vu_\epsilon = \emptyset \qquad \vu_{\omega\,u} = \upshift\vu_\omega\sqcup \vu_u \enspace. \]
To also evaluate goal implications, we further extend $\vu_\omega$ to context variables by using the correspondence defined by $\mapcv$ ($g$ is the output gate of last switch):
\[ \hat\vu_\omega = \vu_\omega \sqcup (\vu_\omega\circ \mapcv^{g}) \enspace. \]
\autoref{thm:sym-concrete-relation} characterizes when a gate value sequence $\omega$ passes or fails a test case derived from a BDDTS $\B$, in terms of $\B$’s execution conditions and goal implications (\autoref{def:goal-implication}). Intuitively, a gate value sequence passes the test case if and only if the execution conditions and goal implications hold for all prefixes of the sequence, up to the point where the test case can no longer proceed—because no continuation satisfies the execution conditions. In this sense, $\omega$ represents a “complete” run of the test. On the other hand, $\omega$ fails the test case if there is a prefix of $\omega$ so that its execution conditions hold, i.e. the path in $\B$ can be taken, but the goal implication at the end is not satisfied.
\begin{restatable}[Correctness of Symbolic Semantics]{theorem}{symbolicvsconcrete}
\label{thm:sym-concrete-relation}
Let $\B$ be a saturated output-rich BDDTS and let $\ini \in \TU(\emptyset)^V$ such that $\sem\ini\models \IG_\B$. For any gate value sequence  $\omega \in \GUU^*$, the following holds:
\begin{enumerate}[topsep=\smallskipamount]
\item $\omega\passes\TC(\B, \ini)$ if and only if there is a $\xi\preceq \omega$ such that $\vu_\xi\models \EC_{\B,\ini}(\sigma_\xi)$ and
\begin{enumerate}
\item for all $\bar\xi\preceq \xi$, $\hat\vu_{\bar\xi}\models \GI_{\B,\ini}(\sigma_{\bar\xi})$ and
\item for all $u\in \GUU$, $\vu_{\xi u}\not\models \EC_{\B,\ini}(\sigma_{\xi u})$.
\end{enumerate}

\item $\omega\fails\TC(\B, \ini)$ if and only if
there is a $\xi\preceq \omega$ such that $\vu_\xi\models \EC_{\B,\ini}(\sigma_\xi)$ and $\hat\vu_\xi \not\models\GI_{\B,\ini}(\sigma_{\xi})$.
\end{enumerate}
\end{restatable}
\noindent
The following encapsulates the last main theoretical contribution of this paper. It follows directly from \autoref{thm:sym-concrete-relation} and the definition of $\simeq$ (\autoref{def:testing-equivalence}).
\begin{restatable}[Testing equivalence preserves test verdicts]{corollary}{corollarytestingequiv}\label{corr:testing-equivalence}
Let $\BB_1,\BB_2$ be sets of saturated output-rich BDDTSs, and let $\ini \in \TU(\emptyset)^V$. If $\BB_1 \simeq \BB_2$, then for all value sequences $\omega\in \GUU^*$:
\begin{align*}
\lefteqn{\exists \B\in \BB_1\restr\ini:\TC(\B,\ini) \fails \omega}\qquad\qquad \\
   & \enspace\Leftrightarrow\enspace \exists \B\in \BB_2\restr\ini:\TC(\B,\ini) \fails\omega \\
\lefteqn{\forall \B\in \BB_1\restr\ini: \TC(\B,\ini) \passes\omega }\qquad\qquad \\
   & \enspace\Leftrightarrow\enspace \forall \B\in \BB_2\restr\ini:\TC(\B,\ini) \passes\omega  \enspace.
\end{align*}
\end{restatable}

\section{Case study: train information boards}
\label{sec:casestudy-result}
\begin{figure*}[t!]
     \centering
     \includegraphics[scale=0.8]{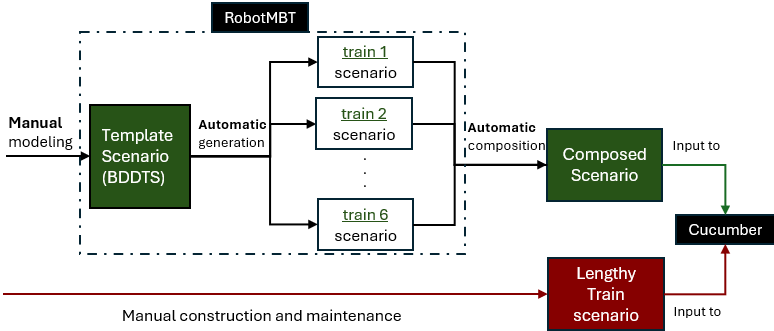}
     \caption{Compositional MBT+BDD approach vs. the NS BDD approach}
     \label{fig:cs-overview}
 \end{figure*}
We applied our theory of disjunction composition to an industrial example from NS, the Dutch Railways (Nederlandse Spoorwegen), with the aim of demonstrating that the disjunction composition process can be automated. We note that this case study was of exploratory nature, and hence not meant as a comprehensive evaluation. 

Our case study is a subsystem from NS, responsible for the information displayed on boards regarding train arrivals and departures.
NS uses BDD scenarios to verify that the system can accurately retrieve and present the stored information from a database on the boards. 
Various combinations of train arrivals and departures can occur at a platform, making it essential to test these combinations thoroughly. To test this level of complexity, NS wrote some rather lengthy scenarios, in the format of \autoref{fig:ns-scenario}.
In particular, the \step{Given} step is quite lengthy, because specific trains need to be each spelled out. With the current manual scenarios, always trains with the same station, departure time, etc. are used. 

\begin{figure}[t!]
\begin{subfigure}{0.62\textwidth}
\begin{scenarios}
    \item \label{scenario:big} Requesting departing trains 
    \begin{steps}
        \Given a departing train at station "AMF" with rideId 1
\And the train has a planned departure time of 08:00
\And the train has a planned departure on platform 1
\And the train has current departure platform 1
\Given $\ldots$ (continues with X other trains)

\When I request the departing trains with station "AMF" and platform "1"
\Then I get a 200
\And I get X+1 trains back
\And the 1st departing train contains departure time "08:00"
\And ...
\When ...(requesting all platforms)
\Then ...(the number of trains back and ordered by time)
    \end{steps}
\end{scenarios}
\caption{Format of manual NS BDD scenarios for several trains}
\label{fig:ns-scenario}
\end{subfigure}
\quad
\begin{subfigure}{\columnwidth}
 \includegraphics[scale=1]{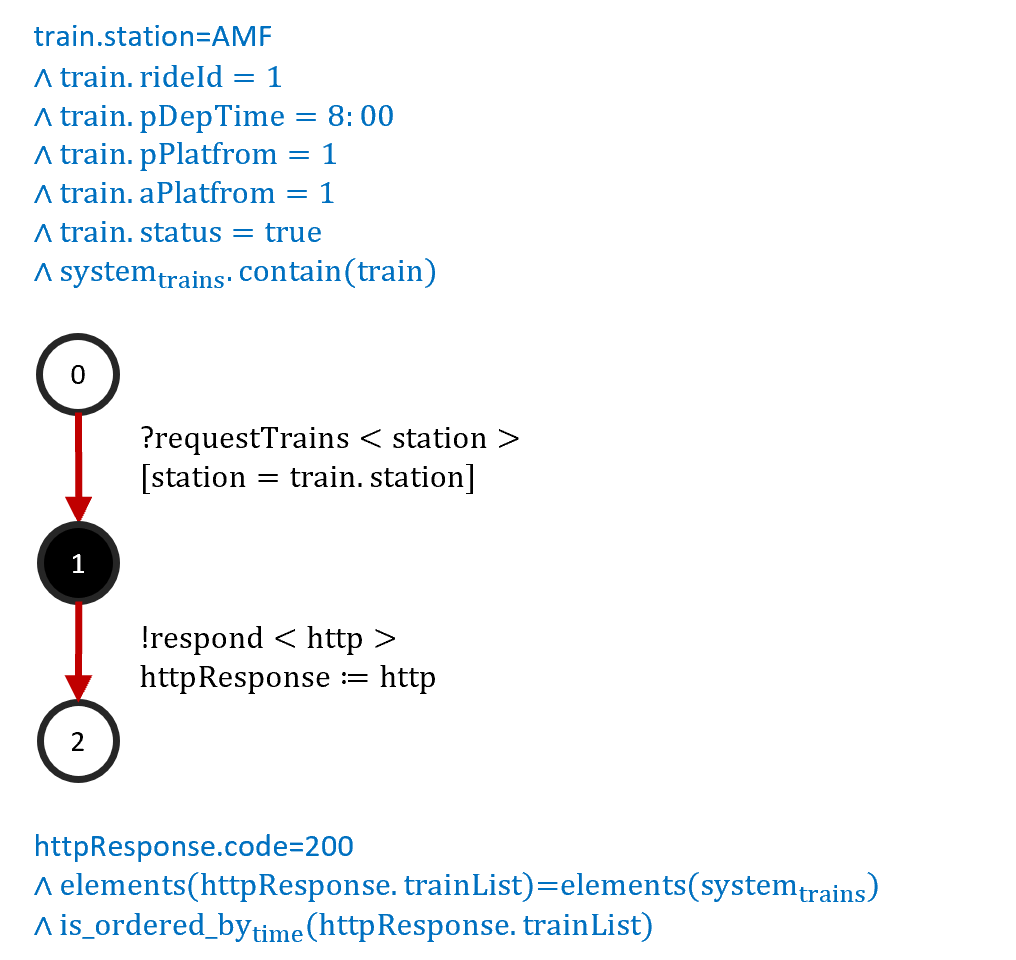}
\caption{BDDTS for one train}
\label{fig:BDDTS-one-train}
\end{subfigure}
\caption{Format of manual BDD scenario and constructed BDDTS}
\label{fig:ns-case-study}
\end{figure}
Therefore there are two issues to address: data diversity and scenario length and complexity. We aim to increase data diversity, to prevent testing the system repeatedly with the same values, e.g. departure time 8:00 and platform “AMF”. We aim to reduce scenario length and complexity of the combinations, to make it easier to read, manage, and update the scenarios, also when requirements evolve.

We address these issues with our approach, as depicted in green in \autoref{fig:cs-overview}; it also shows the NS approach in red.
We first construct a template BDDTS, namely a BDDTS that is an abstract representation of the original scenario. It models the request and response of the current trains on a station, given one train is stored for that station, see \autoref{fig:BDDTS-one-train}. In particular, it models that when the system is requested to return the trains it has stored, the resulting list should contain exactly those trains, ordered by their departure times. In the original BDD scenario (see \autoref{fig:ns-scenario}), this behavior is checked concretely, for example by asserting that a train’s departure time equals a specific value, such as “8:00.”

Next, we create copies of this template BDDTS, where we rename the \code{train} variable to a unique name (train1, train2 etc.) for each copy. Then we apply disjunction composition on those copies to obtain a BDDTS that represents the request and response for multiple trains. Due to the symbolic nature of a BDDTS, we can generate many variations for data values, i.e. for train id, departure time, and platform.

Since no available tool directly implements disjunction composition, we combined an existing library, RobotMBT, and a custom developed script. To be able to use the RobotMBT library as-is, we swapped the data generation and composition: first using RobotMBT, we generated BDD scenarios with concrete data values from our template BDDTS; secondly, we applied disjunction composition on those BDD scenarios with the script. Each of the two is explained in more detail below.

For generating test cases with varied data values we used the open-source Robot Framework library RobotframeworkMBT (introduced in \cite{artifact}). RobotMBT extends the open-source Robot Framework \cite{Robotframework} with model-based testing capabilities. Robot Framework uses keyword-driven testing and supports BDD-style Given–When–Then scenarios, where each step is implemented as a keyword. We expressed our BDDTS of \autoref{fig:BDDTS-one-train} in RobotMBT, by using the modeling functionality of RobotMBT. Then we defined ranges for symbolic values (e.g., range(1,6) for train identifiers, and similar for departure times or platforms) using the tool’s modifier functionality. RobotMBT then automatically generates multiple concrete scenarios from these ranges. Hence, these scenarios represent the copies of the template BDDTS with chosen data values.

We partially implemented disjunction composition, for the BDDTS with data values, by developing a script for extracting  only the `conjunction trace', i.e. the right most branch of the disjunction composition that conjoins the guards of all \code{?requestTrains} switches. This way we obtain the most interesting test case of the disjunction composition, namely the one that combines all trains in a single scenario, which is comparable to the original NS scenario of \autoref{fig:ns-scenario}. We furthermore did some minor conversions, so that the resulting test case had a suitable format for Cucumber, so that it could be executed in that tool, just like the original NS scenario.

We generated test cases with 6 trains, i.e. the number of trains of the NS scenario. Though our data generation was random, we note that it could also generate the values occurring in the original NS scenario.
We executed the generated Cucumber tests 
with different trains until we found a bug. 
The bug was that
the last train was expected to have a departure time of 9:25, while actually a train with departure time 9:10 was stored. With careful debugging, we confirmed that the correct time was indeed saved in the system; instead the testing environment in which the system was running provided the wrong time.

From these exploratory experiments we make the following conclusions. In general, we demonstrated that we were able to automate our approach with disjunction composition. We did so, by building a prototype implementation using existing tooling. With this setup, we were able to generate test cases that combine the request-response commands for information display of any number of trains, instead of the fixed number of trains of the original NS scenario. Also we provided variation of the data values used in the test cases, instead of using the fixed values from the original NS scenario. The result was the discovery of a bug, one that had not been discovered with existing testing methods before. Additionally we note that our template BDDTS is also a lot smaller than the original NS scenario, which would facilitate easier maintenance and updates.

\section{Related Work}
\label{sec:relatedwork}

Below we discuss related work on formal modelling and composition of BDD scenarios, and on composition operators for other transition systems.

\paragraph{Modeling of BDD scenarios.}
De Biase et al. \cite{BDD-to-SysML-Biase2024} generate SysML state machines from gherkin requirements, and Wiecher et al. \cite{SCLL-Wiecher} provide a language in Kotlin for writing BDD scenarios for automotive system validation. Though clearly having the same aim in mind, a key difference is that these works do not employ formal composition. 

\paragraph{Composition of BDD scenarios.}
The authors of \cite{cppBDD} propose representing BDD scenarios for Cyber-Physical Production Systems as a specific kind of state machine and composing those that react to the same event.
Kang et al. translate BDD scenarios into Timed Automata, allowing for manual composition where scenarios share \step{When} steps \cite{TAforBDD}. 
Our disjunction composition also composes shared \step{When} steps. However, both works do not address conflicting \step{Then} steps, like we do.

Other composition operators for BDD scenarios, orthogonal to our paper, are discussed in the following works.
In \cite{DSLforBDD}, a domain-specific language (DSL) is used to abstract from BDD steps. The DSL allows finer-grained scenarios to be embedded within each step, effectively enabling scenario linkage at different abstraction levels.
Formal sequential composition of BDD scenarios is introduced in \cite{Zameni-sequential}. This composition enables subsequent testing of several BDD scenarios.
The SkyFire tool automatically generates BDD scenarios from UML diagrams \cite{Skyfire}. The resulting scenarios tend to be long and are connected in sequence, reflecting the structure of the source UML models. Test cases are then derived using the Cucumber framework. 

\paragraph{Composition operators for transition systems.}
Various composition operators have been studied for Labelled Transition Systems (LTS). In \cite{henzinger-composition}, conjunction, parallel composition, and quotient operators are discussed, with the parallel composition later refined in \cite{henzinger-specifications}. Also, \cite{Janssen-Composition} presents a conjunction operator for LTS
, capturing intersection over outputs. We apply the same idea defined differently via saturation and disjunction, and on BDDTS which is a symbolic model.

 Parallel composition has been studied in the context of ioco-theory in \cite{Bijl-Rensink-Tretmans}. In \cite{Cuyuck-Tretmans} and \cite{gijs-lars-Jan} a notion of mutual acceptance is introduced, such that components can be tested instead of their parallel composition. While in our work the same actions may synchronize, as in parallel composition, we do not consider interleaving; this is an orthogonal concept. Additionally, our BDDTS describe behaviours, i.e. functionality, which may be unrelated to (component) structure of the tested system.

The transformation of STS into LTS is not a core contribution of this paper; instead, we rely on the concept of path conditions from \cite{covsts}. It should be noted that there are alternatives such as \cite{symbolic-execution-explained} that could have been used for this purpose.

Finally, we should acknowledge the huge body of work on composition of LTSs --- essentially, the field of Process Algebra \cite{handbook-process-algebra}. Though this is clearly the source of inspiration for our work, it is by lifting the concepts to the symbolic level where the contribution of this paper lies.

\section{Conclusion and Future Work}
\label{sec:conclusion}

In this work, we present  the
disjunction composition operator for BDD Transition Systems (BDDTS)—formal models of BDD scenarios. This operator enables the combination of multiple BDDTSs 
so that tests can be derived from this integrated model, instead of only testing the scenarios in isolation.
We define a symbolic semantics for BDDTS, to show correspondence between two BDDTS and their composition. We also show the relation between this symbolic semantics and the concrete semantics in terms of test cases.

In the scope of this paper, we have ignored certain aspects of testing such as the analysis of the (root) cause of a failure. Though a failed test can clearly be traced back to the BDDTS path, this is not enough when the BDDTS itself is obtained through (disjunction) composition: one would also want to trace it back to the original BDD scenario that was not fulfilled. However, we expect no particular difficulties here, as every composed BDDTS switch is derived from the component switches through the derivation rules in \autoref{def:disjunction}. It is straightforward to build this dependency into the tooling.

For future work, we first aim to address the fact that saturation can change pass verdicts, and investigate techniques to preserve them. We would also like to develop more composition operators for BDDTS, e.g. for loops and parallel composition. 
Furthermore, we would like to generalize our symbolic semantics to more generic models than BDDTS.
Finally, we aim to do more elaborate case studies, and also further integrate the disjunction composition in MBT tooling.



\bibliography{lib}

\appendix
\section{Proofs of \autoref{sec:disjunction}}
\saturatedcomposition*
\begin{proof}
    Let $\B_i=\bddtstupleind{i}$ be compatible saturated BDDTSs for $i=1,2$ and let $\B=\B_1 \bigtriangledown \B_2=\bddtstuple$ be their composition.
Also, let $(l_1,l_2)\in L$ be an arbitrary location of $\discomp$ and $\alpha \in \IAU(G)$ an interaction. $T=\{t \mid sl_{t} =(l_1,l_2), \alpha_{t} =\alpha\} \subseteq \arrow$ denotes the set of all outgoing switches from $(l_1,l_2)$ with interaction $\alpha$.

\medskip\noindent
\textbf{First condition}: We show that $\bigvee_{t \in T}\phi_t \equiv \True$
\begin{enumerate}
    \item as $\B_1$ and $\B_2$ are saturated, we know that for $T_1=\{t \mid sl_{t} =l_1, \alpha_{t} =\alpha\} \subseteq \arrow_1$ (the set of outgoing switches from $l_1$ with $\alpha$) we have that $\bigvee_{t \in T_1}\phi_t \equiv \True$ and similarly for $T_2=\{t\mid sl_{t} =l_2, \alpha_{t} =\alpha\} \subseteq \arrow_2$, $\bigvee_{t \in T_2}\phi_t \equiv \True$
    \item 
    The switches in $T$ can be constructed by rule 1 or rule 2.
    \item Any $t \in T$ constructed by rule 1, has the shape: $t=((l_1,l_2), \phi_1 \land \phi_2, a_1 \sqcup a_2,(l_1',l_2'))$ so we can write:
    \item $\forall t \in T, t_1\in T_1, t_2 \in T_2, \phi_t=\phi_{t_1} \land \phi_{t_2}$
    \item as rule 1 constructs the conjunction of $\phi$s $\bigvee_{t \in T}\phi_t= \bigvee_{t_1 \in T_1}\bigvee_{t_2 \in T_2} \phi_{t_1} \land \phi_{t_2} $ by de Morgan's law.$=\bigvee_{t_1\in T_1} \left(\phi_{t_1} \land \bigvee_{t_2\in T_2} \phi_{t_2}\right)=\True \land \True=\True$ 
    
    \item for any $t \in T$ If it is constructed by rule 2-1, it has the shape: $t=((l_1,l_2), \phi_1 , a_1 ,(l_1',\bot))$ then $\bigvee_{t \in T}\phi_t= \bigvee_{t_1 \in T_1} \phi_{t_1} =\True$
    \item for any $t \in T$ s constructed by rule 2-2, the proof is similar to 2-1.
    
\end{enumerate}
\textbf{Second condition}: We show that for any $(l_1, l_2) \in L^\Cl$ and $\alpha \in \IAU(G_o)$, there is a switch $(l_1, l_2) \trans{\alpha,\phi,a}_\B$.
\begin{enumerate}
    \item As BDDTSs $\B_1$ and $\B_2$ are saturated, we know that for any $ l_1\in L_1^\Cl$ and $\alpha \in \IAU(G_o)$, there is a switch $l_1 \trans{\alpha,\phi_1,a_1}_{\B_1}$ and  for any $ l_2\in L_2^\Cl$ and $\alpha \in \IAU(G_o)$, there is a switch $l_2 \trans{\alpha,\phi_2,a_2}_{\B_2}$.
    \item By the definition of composition with either rule 1 or rule 2, all the switches in $\B_1$ and $\B_2$ appear in the composition as stated below:
    \begin{itemize}
        \item By rule 1: for all $l_1 \in L_1^\Cl$ where $l_1 \trans{\alpha,\phi_1,a_1}_{\B_1}$  and $\alpha \in \IAU(G_o)$ and for any $ l_2\in L_2^\Cl$ and  $l_2 \trans{\alpha,\phi_2,a_2}_{\B_2}$  then first by definition we have$(l_1,l_2) \in L^\Cl$ and by rule 1 $(l_1,l_2) \trans{\alpha,\phi_1 \land \phi_2,a_1 \sqcup a_2}$
        \item by rule 2-1 for all $l_1 \in L_1^\Cl$ where $l_1 \trans{\alpha,\phi_1,a_1}_{\B_1}$  and $\alpha \in \IAU(G_o)$ and $l_2 \ntrans{\alpha}$, $(l_1,l_2) \trans{\alpha,\phi_1 ,a_1}$. rule 2-2 is similar to this case
    \end{itemize}
    and the condition follows from above items 1 and 2.
\end{enumerate}
\textbf{Third condition}: We show that for any $(\il_1,\il_2)\trans{\alpha,\phi,a} (l_1,l_2)$, either $\phi\Rrightarrow \IG_1 \lor \IG_2$ or $(l_1,l_2)\in L_{\B,\sink}^\Op$ and $\OG\undefinedon l$.
\begin{enumerate}
    \item As BDDTSs $\B_1$ and $\B_2$ are saturated, we know that for all $\il_1\trans{\alpha,\phi_1,a_1} l_1$ $\phi_1\Rrightarrow \IG_1 $ or $l_1\in L_{\B_1,\sink}^\Op$ and similarly in $\B_2$ we have for any $\il_2\trans{\alpha,\phi_2,a_2} l_2$ $\phi_2\Rrightarrow \IG_2 $ or $l_2\in L_{\B_2,\sink}^\Op$ and $\OG_1 \uparrow l_1$ and $\OG_2 \uparrow l_2$ 
    \item There are 4 cases (stated below) when rule 1 is applied, for any $\il_1\trans{\alpha,\phi_1,a_1} l_1$ and $\il_2\trans{\alpha,\phi_2,a_2} l_2$ 
    \begin{enumerate}
        \item we have $\phi_1\Rrightarrow \IG_1 $ and $\phi_2\Rrightarrow \IG_2 $ and by rule 1 $(\il_1,\il_2) \trans{\alpha,\phi_1 \land \phi_2,a_1 \sqcup a_2}$
        \item [] as $\phi_1 \Rrightarrow \IG_1$ and $\IG=\IG_1 \lor \IG_2$ by disjunction introduction it holds that $\phi_1 \implies \IG_1 \lor \IG_2$
        \item [] similarly $\phi_2 \implies \IG_1 \lor \IG_2$
        \item [] By conjunction introduction and Monotonicity of implication $\phi_1 \land \phi_2 \implies IG_1 \lor \IG_2$ holds.
        \item $\phi_1\Rrightarrow \IG_1 $ and $l_2\in L_{\B_2,\sink}^\Op$
        \item [] Here we don't know about $\phi_2$ but with above reasoning and monotonicity of implication $\phi_1 \Rightarrow \IG_1 \lor \IG_2$ so adding $\phi_2$ does not make a difference and again  $\phi_1 \land \phi_2 \implies \IG_1 \lor \IG_2$ holds.
        \item $\phi_2\Rrightarrow \IG_2 $ and $l_1\in L_{\B_1,\sink}^\Op$
        \item [] with the same reasoning as previous item, this also holds.
        \item $l_1\in L_{\B_1,\sink}^\Op$ and $l_2\in L_{\B_2,\sink}^\Op$ so the composed location in composition will be an open sink location: $(l_1,l_2) \in L_{{\discomp}, \sink}^\Op$ and by definition as none of the locations have $\OG$ defined then it holds that $\OG \uparrow (l_1,l_2)$
    \end{enumerate}
    \item If rule 2-1 is applied and  $\phi_1\Rrightarrow \IG_1 $ then only $\phi_1$ appears as the switch guard in the composition and it immediately follows that $\phi_1 \Rrightarrow \IG_1 \lor \IG_2$ and if  $l_1\in L_{\B_1,\sink}^\Op$ then in the composition by rule 2 it ends up to location $(l_1, \bot)$ which is both open and sink without $\OG$ by definition.
    \item Rule 2-2 is symmetric to 2-1.
\end{enumerate}
\end{proof}

\subsection{Commutativity and associativity of disjunction composition}

For the proof of \autoref{prop:disjunctioncommassoc}, we use isomorphism over BDDTSs. For completeness, this (known) concept is therefore defined below.

\begin{definition} [BDDTS Isomorphism] \label{isomorphism}
Two BDD Transition Systems are isomorphic if their structure is the same.
Let  $B_1= \bddtstupleind{1} $ and $B_2= \bddtstupleind{2}$ be two  BDDTSs. They are isomorphic iff:
\end{definition}
 There is a bijective function $f$ mapping  locations from $L_1$ to $L_2$ preserving their natures such that:
    \begin{itemize}
        \item $f(il_1)=il_2$ 
         \item and $\IG_1 \equiv \IG_2 $
         \item $l \xrightarrow{\lambda_1}_1 l'$ for some $l,l'\in L_1$ if and only if $f(l) \xrightarrow{\lambda_2}_2 f(l')$
         such that  switch labels are equivalent $\lambda_1=\lambda_2$ for $\lambda_1 =(\alpha_1, \phi_1, a_1)$ and $\lambda_2 =(\alpha_2, \phi_2, a_2)$ where equivalence means:
         \begin{itemize}
             \item The gates are equal: $\alpha_1=\alpha_2$ which follows that $g_1=g_2$ and  $\ivbar_1=\ivbar_2$
            \item The switch guards are semantically equivalent $\phi_1 \equiv \phi_2$
            \item The assignments are compatible: $a_1 \compat a_2$
        
    \end{itemize}
    \item  $\OG \downarrow l_1$ iff $\OG \downarrow l_2$ and $\OG(f(l_1)) \equiv \OG(l_2)$
   
\end{itemize}

\disjunctioncommassoc*

\begin{proof}~

\noindent
\textbf{Commutativity.}

\medskip\noindent
We show that the BDDTS $\B_1 \bigtriangledown \B_2$ is isomorphic to the BDDTS $\B_2  \bigtriangledown \B_1$, by proving that the bijective mapping function:$ \isomap(l_1,l_2)=(l_2,l_1)$ defines this isomorphism. Let $\lambda \in \Lambda_{\discomp}$ be an arbitrary switch label in the composition.

\begin{enumerate}
    \item  \textbf{Case 1}: Suppose that $\lambda$ was constructed with rule 1, i.e. $(l_1,l_2) \trans{\alpha,\phi_1 \land  \phi_2,a_1 \sqcup a_2} (l_1',l_2')$. Then we have:
   \begin{itemize}
       \item $l_2 \xrightarrow{\alpha,\phi_2,a_2} l_2' $
       and  $l_1 \xrightarrow{\alpha,\phi_1,a_1} l_1'$
        \item If $l_1=\il_1$ and $l_2=\il_2$ are initial locations: $\IG_{\discomp}=\IG_1 \lor \IG_2$ and $\IG_{\B_2 \bigtriangledown \B_1}= \IG_2 \lor \IG_1$, so $\IG_{\discomp} \equiv \IG_{\B_2 \bigtriangledown \B_1} $\\
   \end{itemize}
   \item [] We have that $\isomap(l_1,l_2)=(l_2,l_1)$ in $B_2 \bigtriangledown B_1$, rule 1 is applicable, i.e. 
   \begin{itemize}
       \item $(l_2,l_1) \trans{\alpha,  \phi_1 \land \phi_2, a_1 \sqcup a_2} (l'_2, l'_1)$
       \item and $\phi_2 \land \phi_1 \equiv \phi_1 \land \phi_2$
       \item and by definition $a_1 \sqcup a_2 \compat a_2 \sqcup a_1$.
        \item if $\OG \downarrow l_1$ and $\OG \downarrow l_2$, by definition: 
        \begin{itemize}
            \item $\OG_{\discomp}= \OG_1(l_1) \land \OG_2(l_2)$
            \item $\OG_{\B_2 \bigtriangledown \B_1}= \OG_2(l_2) \land \OG_1(l_1)$
            \item Thus: $\OG_{\discomp} \equiv \OG_{\B_2 \bigtriangledown \B_1}$
            \item The same applies for $l_1'$ and $l_2'$
         \end{itemize}
          \item if $\OG \downarrow l_1$ and $\OG \uparrow l_2$, by definition: 
        \begin{itemize}
            \item $\OG_{\discomp}= \OG_1(l_1) $
            \item $\OG_{\B_2 \bigtriangledown \B_1}= \OG_1(l_1)$
            \item Thus: $\OG_{\discomp} \equiv \OG_{\B_2 \bigtriangledown \B_1}$
            \item The same applies for $l_1'$ and $l_2'$
         \end{itemize}
         \item if $\OG \uparrow l_1$ and $\OG \downarrow l_2$, by definition: 
        \begin{itemize}
            \item $\OG_{\discomp}= \OG_2(l_2) $
            \item $\OG_{\B_2 \bigtriangledown \B_1}= \OG_2(l_2)$
            \item Thus: $\OG_{\discomp} \equiv \OG_{\B_2 \bigtriangledown \B_1}$
            \item The same applies for $l_1'$ and $l_2'$
         \end{itemize}
   \end{itemize}

   Consequently, $\isomap(l_1', l_2')= (l_2' , l_1')$
 
    \item  \textbf{Case 2}: Suppose that $\lambda$ was constructed with rule 2-1, i.e. $(l_1,l_2) \xrightarrow{\alpha,\phi,a_1 } (l_1',\bot)$.
    Then we have:
    \begin{itemize}
        \item  $l_1 \xrightarrow{\alpha,\phi,a_1} l_1'$
        \item $l_2 \ntrans{\alpha}$ no switch enabled from $l_2$ with interaction $\alpha$ in $\B_2$
    \end{itemize}
    Hence, we have that for location $\isomap(l_1,l_2) = (l_2,l_1)$ in $\B_2 \bigtriangledown \B_1$, rule 2-2 is applicable, i.e. 
    \begin{itemize}
        \item $(l_2,l_1) \xrightarrow{\alpha,\phi_1,a_1 } (\bot,l_1')$.
        \item $\phi_1 \equiv \phi_1$
        \item $a_1 \compat a_1$
        \item For $\IG$ and $\OG$ the same reasoning as the first case is applicable. Note that $l_2'=\bot$ so by default $\OG \uparrow l_2'$ 
    \end{itemize} Consequently, we have $\isomap(l_1',\bot)=(\bot,l_1')$ as required. 
    \item  \textbf{Case 3}: Suppose that $\lambda$ was constructed with rule 2-2. The the proof is symmetric to the previous case, since rules 2-1 and 2-2 are symmetric.
\end{enumerate}
\textbf{Associativity.}

\medskip\noindent
We prove this by showing that the resulting composed BDDTSs $B_3 \bigtriangledown (B_1 \bigtriangledown B_2)$ and $(B_3  \bigtriangledown B_1) \bigtriangledown B_2$ are isomorphic with the bijective function:\\
$
\isomap(l_3,(l_1,l_2))=((l_3,l_1),l_2)
$ 
\begin{itemize}
\item  \textbf{Case 1}: $\lambda_{\discomp}$ was constructed with rule 1, i.e. $(l_1,l_2) \trans{\alpha,\phi_1 \land  \phi_2,a_1 \sqcup a_2} (l_1',l_2')$. and $\lambda_{\discompthree}$ is also constructed with rule 1,  i.e. $(l_3,(l_1,l_2)) \trans{\alpha,\phi_3 \land (\phi_1 \land  \phi_2),a_3 \sqcup (a_1 \sqcup a_2)} (l_3',(l_1',l_2'))$.

Then we have:
   \begin{itemize}
       \item $(l_1,l_2) \xrightarrow{\alpha,\phi_1 \land \phi_2,a_1 \sqcup a_2} (l_1',l'_2)$
       and  $l_3 \xrightarrow{\alpha,\phi_1,a_1} l_3'$
        \item If $l_1=\il_1$, $l_2=\il_2$ and $l_3=\il_3$ are initial locations: $\IG_{\discompthree}=\IG_3 \lor (\IG_1 \lor \IG_2)$ and $\IG_{\discomptwoone}= (\IG_3 \lor \IG_1) \lor \IG_2$, so $\IG_{\discompthree} \equiv \IG_{\discomptwoone} $\\
   \end{itemize}
   \item [] Hence we have that for location $\isomap(l_3,(l_1,l_2))=((l_3,l_1),l_2)$ in $(\B_3 \bigtriangledown \B_1) \bigtriangledown \B_2$, rule 1 is applicable, i.e. 
   \begin{itemize}
       \item In $\B_3 \bigtriangledown B_1$ the result of rule 1 is: $(l_3,l_1) \trans{\alpha,  \phi_3 \land \phi_1, a_3 \sqcup a_1} (l'_3, l'_1)$
       \item We also have: $l_2 \xrightarrow{\alpha,\phi_2,a_2} l_2'$
       \item by applying rule 1 to build $\discomptwoone$ we have: $((l_3,l_1),l_2) \trans{\alpha,  (\phi_3 \land \phi_1)\land \phi_2, (a_3 \sqcup a_1) \sqcup a_2} ((l'_3, l'_1),l'_2)$
       \item and $(\phi_3 \land \phi_1) \land \phi_2 \equiv \phi_3 \land (\phi_1 \land \phi_2)$
       \item and by definition $(a_3 \sqcup a_1) \sqcup a_2 \compat a_3 \sqcup (a_1 \sqcup a_2) $.
        \item if $\OG \downarrow l_1$ and $\OG \downarrow l_2$, and  $\OG \downarrow l_3$  by definition: 
        \begin{itemize}
            \item $\OG_{\discompthree}= \OG_3(l_3)\land(\OG_1(l_1) \land \OG_2(l_2))$
            \item $\OG_{\discomptwoone}= (\OG_3(l_3) \land \OG_1(l_1)) \land \OG_2(l_2)$
            \item Thus: $\OG_{\discompthree} \equiv \OG_{\discomptwoone}$
            \item The same applies for $l_1'$,$l_2'$ and $l_3'$
         \end{itemize}
          \item if any of $\OG$s are undefined for any of the locations, the corresponding $\OG$ will be removed in the above expression and the left hand side and the right hand side still stay equivalent.
   \end{itemize}   Consequently, $\isomap(l_3',(l_1', l_2'))= ((l_3' , l_1'),l_2')$
\item  \textbf{Case 2}: $\lambda_{\discomp}$ was constructed with rule 2-1, i.e. $(l_1,l_2) \trans{\alpha,\phi_1 ,a_1 } (l_1',\bot)$. and $\lambda_{\discompthree}$ is constructed with rule 1,  i.e. $(l_3,(l_1,l_2)) \trans{\alpha,\phi_3 \land \phi_1, a_3 \sqcup a_1 } (l_3',(l_1',\bot))$.
Then we had:
   \begin{itemize}
       \item In $\discomp$: $(l_1,l_2) \xrightarrow{\alpha,\phi_1 ,a_1 } (l_1',\bot)$
        which means that $l_2 \ntrans{\alpha}$
       \item We also have: $l_3 \xrightarrow{\alpha,\phi_3,a_3} l_3'$
   \end{itemize}
   \item [] Hence we have that for location $\isomap(l_3,(l_1,l_2))=((l_3,l_1),l_2)$ in $(\B_3 \bigtriangledown \B_1) \bigtriangledown \B_2$, rule 2-1 is applicable, i.e. 
    \begin{itemize}
       \item $(l_3,l_1) \trans{\alpha,  \phi_3 \land \phi_1, a_3 \sqcup a_1} (l'_3, l'_1)$
       and  $l_2 \ntrans{\alpha}$
       \item by applying rule 2-1 in: $\discomptwoone$ $((l_3,l_1),l_2) \trans{\alpha,  \phi_3 \land \phi_1, a_3 \sqcup a_1 } ((l'_3, l'_1),\bot)$
        \item Comparing $\lambda_{\discompthree}$ and $\lambda_{\discomptwoone}$:
       \begin{itemize}
           \item $\alpha=\alpha$
           \item $\phi_3 \land \phi_1 \equiv \phi_3 \land \phi_1 $
            \item and by definition $a_3 \sqcup a_1\compat a_3 \sqcup a_1 $.
       \end{itemize}   
   \end{itemize}   Consequently, $\isomap(l_3',(l_1', \bot))= ((l_3' , l_1'),\bot)$

\item  \textbf{Case 3}: $\lambda_{\discomp}$ was constructed with rule 2-2, i.e. $(l_1,l_2) \trans{\alpha,\phi_2 ,a_2 } (\bot,l_2')$. and $\lambda_{\discompthree}$ is constructed with rule 1,  i.e. $(l_3,(l_1,l_2)) \trans{\alpha,\phi_3 \land \phi_2, a_3 \sqcup a_2 } (l_3',(\bot,l_2'))$.
Then:
   \begin{itemize}
       \item in $\discomp$ we had $(l_1,l_2) \trans{\alpha,\phi_2 ,a_2 } (\bot,l_2')$
        which means that $l_1 \ntrans{\alpha}$
       \item we also have: $l_3 \xrightarrow{\alpha,\phi_3,a_3} l_3'$
       \item so  by applying rule 1 in $\discompthree$:  $(l_3,(l_1,l_2)) \trans{\alpha,\phi_3 \land \phi_2, a_3 \sqcup a_2 } (l_3',(\bot,l_2'))$
       
   \end{itemize}
   \item [] Hence we have that for location $\isomap(l_3,(l_1,l_2))=((l_3,l_1),l_2)$
    \begin{itemize}
       \item In $\B_3 \bigtriangledown \B_1$ rule 2-1 should have been applied (because $l_1 \ntrans{\alpha}$):
        $(l_3,\bot) \trans{\alpha,  \phi_3 , a_3 } (l'_3,\bot)$
       \item we also had: $l_2 \trans{\alpha, \phi_2,a_2} l_2'$
       \item  in $(\B_3 \bigtriangledown \B_1) \bigtriangledown \B_2$, rule 1 is applicable, i.e. 
       \item [] $((l_3,l_1),l_2) \trans{\alpha,  \phi_3 \land \phi_2, a_3 \sqcup a_2 } ((l'_3, \bot),l_2')$
       \item Comparing $\lambda_{\discompthree}$ and $\lambda_{\discomptwoone}$:
       \begin{itemize}
           \item $\alpha=\alpha$
           \item $\phi_3 \land \phi_2 \equiv \phi_3 \land \phi_2 $
           \item by definition $a_3 \sqcup a_2\compat a_3 \sqcup a_2 $.
       \end{itemize}
   \end{itemize}   Consequently, $\isomap(l_3',(\bot, l_2'))= ((l_3' , \bot),l_2')$
   
   \item  \textbf{Case 4}: $\lambda_{\discomp}$ was constructed with rule 1, i.e. $(l_1,l_2) \trans{\alpha,\phi_1 \land \phi_2 ,a_1 \sqcup a_2 } (l_1',l_2')$. and $\lambda_{\discompthree}$ is constructed with rule 2-2,  i.e. $(l_3,(l_1,l_2)) \trans{\alpha,\phi_1 \land \phi_2, a_1 \sqcup a_2 } (\bot,(l_1',l_2'))$.
   Then we have:
   \begin{itemize}
       \item $(l_1,l_2) \xrightarrow{\alpha,\phi_1 \land \phi_2,a_1 \sqcup a_2} (l_1',l'_2)$
       and  $l_3 \ntrans{\alpha}$
   \end{itemize}
   \item [] Hence we have that for location $\isomap(l_3,(l_1,l_2))=((l_3,l_1),l_2)$ in $(\B_3 \bigtriangledown \B_1) \bigtriangledown \B_2$, rule 1 is applicable, i.e. 
    \begin{itemize}
       \item As $l_3 \ntrans{\alpha}$, in $\B_3 \bigtriangledown \B_1$ rule 2-2 has been applied: $(l_3,l_1) \trans{\alpha, \phi_1,  a_1} (\bot, l'_1)$
       and  $l_2 \trans{\alpha, \phi_2, a_2} l_2'$
       \item by applying rule 1 in $\discomptwoone$:  $((l_3,l_1),l_2) \trans{\alpha,  \phi_1 \land \phi_2, a_1 \sqcup a_2 } ((\bot, l'_1),l_2')$
        \item Comparing $\lambda_{\discompthree}$ and $\lambda_{\discomptwoone}$:
       \begin{itemize}
           \item $\alpha=\alpha$
           \item $\phi_1 \land \phi_2 \equiv \phi_1 \land \phi_2 $
            \item and by definition $a_1 \sqcup a_2\compat a_1 \sqcup a_2 $.
       \end{itemize}   
   \end{itemize}   Consequently, $\isomap(\bot',(l_1',l_2' ))= ((\bot , l_1'),l_2')$
    \item  \textbf{Case 5}: We had that $l_1 \ntrans{\alpha}$ and $l_2 \ntrans{\alpha}$ and $\lambda_{\discompthree}$ is constructed with rule 2-2,  i.e. $(l_3,(l_1,l_2)) \trans{\alpha,\phi_3, a_3 } (l_3',(\bot,\bot))$.
  
   \item [] Hence we have that for location $\isomap(l_3,(l_1,l_2))=((l_3,l_1),l_2)$ in $(\B_3 \bigtriangledown \B_1) \bigtriangledown \B_2$, rule 2-2 is applicable, i.e. 
    \begin{itemize}
       \item As $l_1 \ntrans{\alpha}$, in $\B_3 \bigtriangledown \B_1$ rule 2-1 has been applied: $(l_3,l_1) \trans{\alpha, \phi_3,  a_3} (l_3', \bot)$
       and  $l_2 \trans{\alpha, \phi_2, a_2} l_2'$
       \item by applying rule 2-1 in $\discomptwoone$:  $((l_3,l_1),l_2) \trans{\alpha,  \phi_3, a_3 } ((l_3', \bot),\bot)$
        \item Comparing $\lambda_{\discompthree}$ and $\lambda_{\discomptwoone}$:
       \begin{itemize}
           \item $\alpha=\alpha$
           \item $\phi_3 \equiv \phi_3 $
            \item and by definition $a_3\compat a_3 $.
       \end{itemize}   
   \end{itemize}   Consequently, $\isomap(l_3',(\bot,\bot))= ((l_3' , \bot),\bot)$

    \item  \textbf{Case 6}: We had that $l_3 \ntrans{\alpha}$ and $l_1 \ntrans{\alpha}$ and $\lambda_{\discompthree}$ is constructed with rule 2-2,  i.e. $(l_3,(l_1,l_2)) \trans{\alpha,\phi_2, a_2 } (\bot,(\bot,l_2'))$.
   Then we have:
   \begin{itemize}
       \item $(l_1,l_2) \xrightarrow{\alpha,\phi_2,a_2} (\bot,l'_2)$
       and  $l_3 \ntrans{\alpha}$
   \end{itemize}
   \item [] Hence we have that for location $\isomap(l_3,(l_1,l_2))=((l_3,l_1),l_2)$ in $(\B_3 \bigtriangledown \B_1) \bigtriangledown \B_2$, rule 2-2 is applicable, i.e. 
    \begin{itemize}
       \item As $l_1 \ntrans{\alpha}$, 
       and $l_3 \ntrans{\alpha}$ no rule is applied in $\B_3 \bigtriangledown \B_1$
       \item by applying rule 2-2 in $\discomptwoone$:  $((l_3,l_1),l_2) \trans{\alpha,  \phi_2, a_2 } ((\bot, \bot),l_2)$
        \item Comparing $\lambda_{\discompthree}$ and $\lambda_{\discomptwoone}$:
       \begin{itemize}
           \item $\alpha=\alpha$
           \item $\phi_2 \equiv \phi_2 $
            \item and by definition $a_2\compat a_2 $.
       \end{itemize}   
   \end{itemize}   Consequently, $\isomap(\bot,(\bot,l_2'))= ((\bot , \bot),l_2')$
   
\end{itemize}

All cases are covered and the isomorphism holds.
\end{proof}
\subsection{Test Behavior Preservation in Disjunction Composition}

To prove the theorems in the paper we need to define a couple of concepts and notations described in this section.

We define the notion of label injection: For $i = 1, 2$, the function $\Proj[i]{l}$ returns the set of pairs in $L$ where the label $l$ appears in the $i$-th position:
\[ \Proj[i]{l}=\gensetof{(l_1,l_2) \in L}{l_i=l}
\]
This identifies all composite labels in which $l$ is injected in position $i$.

For two partial functions $f,g\in Y^X_\bot$, we write $f\ssby g$ (``$f$ is subsumed by $g$'') if, for all $x\in X$, $f\definedon x$ implies $g\definedon x$ and $g(x)=f(x)$. This is extended to paths: we define ${\ssby}\subseteq \Lambda^*\times \Lambda^*$  as the smallest relation such that
\begin{itemize}
\item $\epsilon \ssby \epsilon$
\item If $\pi_1\ssby \pi_2$, $\phi_1\Leftarrow \phi_2$ and $a_1\ssby a_2$, then $\pi_1\cdot(\alpha,\phi_1,a_1)\ssby \pi_2\cdot(\alpha,\phi_2,a_2)$.
\end{itemize}
Hence, $\pi_1\ssby \pi_2$ (``$\pi_1$ is subsumed by $\pi_2$'') if the elements $(\alpha_1,\phi_1,a_1)$ in $\pi_1$ correspond on a one-to-one basis with the elements $(\alpha_2,\phi_2,a_2)$ of $\pi_2$ in the sense that $\alpha_1=\alpha_2$, $\phi_1\Leftarrow \phi_2$ and $a_1\ssby a_2$. The intuition is that, when composing $\B_i$ for $i=1,2$, all paths of each $\B_i$ to some location $l_i$ are subsumed by paths of $\B_1\bigtriangledown \B_2$ to $\Proj[i]l$. To formulate this precisely, we also need the notion of (subsuming) \emph{location path}.
\begin{definition}[Location Path]\label{def:path}
Let $\B$ be a BDDTS. The \emph{location path} function $LP_\B\colon \IAU^*\times L_\B \to 2^{\Lambda_\B}$ returns the paths that reach a location $l$ via an interaction sequence $\sigma$:
\[ \PCB{\B}{\sigma}{l} =
  \{\pi \mid \il_\B\trans\pi_\satK l,\; \sigma_\pi = \sigma \} \enspace.
\]
We also use the subset of location paths subsuming a given path $\pi$:
\[ LP_\B(\sigma,l\mid \pi) = \{\varpi\in LP_\B(\sigma,l) \mid \pi\ssby \varpi\}  \enspace.
\]
\end{definition}
The below lemma, states that for any path in the composition there is at least one path in one of the BDDTSs that is subsumed by the path in the composition:
\begin{lemma}[path subsumption]\label{lemma:path-subsumption}
Let $\B_i = \bddtstupleind{i}$ for $i=1,2$ be compatible saturated BDDTSs and let $\B = \B_1 \bigtriangledown \B_2=\bddtstuple$. Let $\ini \in \TU(\emptyset)^V$. For any $\sigma \in \IAU^*$ and any $\pi \in \PCB{\B}{\sigma}{l}$, there exists $\pi_i \in \PCB{\B_i}{\sigma}{l_i}$ such that $\pi_i \ssby \pi$.
\end{lemma}
\begin{proof}[Induction on the length of $\pi$]
We prove the following statement by induction on $n \geq 0$:
\[
P(n):
\text{For all } \sigma \in \IAU^*,\ \text{and for all } \pi \in \PCB{\B}{\sigma}{l}\ \text{with } |\pi|=n,\]\[
\text{there exists}\ \pi_i \in \PCB{\B_i}{\sigma}{l_i}\ \text{such that } \pi_i \ssby \pi.
\]

\noindent
\textbf{Base case ($n=0$).}  
Suppose $\pi \in \PCB{\B}{\sigma}{l}$ with $|\pi|=0$. Then $\pi$ is the empty path.  
Let $\pi_i$ be the empty path in $\PCB{\B_i}{\sigma}{l_i}$. Clearly $\pi_i \ssby \pi$.  
Thus $P(0)$ holds.

\medskip\noindent
\textbf{Induction hypothesis.}  
Assume $P(n)$ holds for some $n \geq 0$, i.e.,
every $\pi' \in \PCB{\B}{\sigma'}{l'}$ of length $n$ admits 
$\pi'_i \in \PCB{\B_i}{\sigma'}{l_i'}$ such that $\pi_i' \ssby \pi'$.

\medskip\noindent
\textbf{Induction step ($n \Rightarrow n+1$).}  
Let $\pi \in \PCB{\B}{\sigma}{l}$ with $|\pi|=n+1$.  
We write $\pi = \pi' \cdot \lambda$ where $\pi'$ is the prefix of length $n$ and $\lambda$ is the label of the last switch $t$ where $\lambda=\lambda_t$.

\begin{enumerate}
    \item By construction, $\pi' \in \PCB{\B}{\sigma'}{l'}$ with $|\pi'|=n$.  
    By the induction hypothesis, there exists $\pi'_i \in \PCB{\B_i}{\sigma'}{l_i'}$ such that $\pi_i' \ssby \pi'$.
   
    \begin{itemize}
        \item If rule 1 of disjunction is applied on  $t_1 \in \arrow_1$ and $t_2 \in \arrow_2$ we have: $\alpha_{\lambda}=\alpha_{t_1}=\alpha_{t_2}$, $\phi_{\lambda}= \phi_{t_1}\land \phi_{{t_2}}$ and $a_{\lambda}= a_{t_1} \sqcup a_{{t_2}}$
        \item If rule 2-1 of disjunction is applied on  $t_1 \in \arrow_1$ and $t_2 \in \arrow_2$ we have: $\alpha_{\lambda}=\alpha_{{t_1}}$, $\phi_{\lambda}= \phi_{{t_1}}$ and $a_{\lambda}= a_{{t_1}}$
        \item If rule 2-2 of disjunction is applied on  $t_1 \in \arrow_1$ and $t_2 \in \arrow_2$ we have: $\alpha_{\lambda}=\alpha_{{t_2}}$, $\phi_{\lambda}= \phi_{{t_2}}$ and $a_{\lambda}= a_{{t_2}}$
    \end{itemize}
    \item All the three cases above satisfy the subsumption condition $\alpha_{\lambda}=\alpha_{{t_i}}$, $\phi_{t_i}\Leftarrow \phi_{\lambda}$ and $a_{t_i}\ssby a_{\lambda}$.
\end{enumerate}

Hence $P(n+1)$ holds.

\medskip
\noindent
\textbf{Conclusion.}  
By induction, $P(n)$ holds for all $n \geq 0$.  
Therefore, for any $\sigma \in \IAU^*$ and any $\pi \in \PCB{\B}{\sigma}{l}$, there exists $\pi_i \in \PCB{\B_i}{\sigma}{l_i}$ such that $\pi \ssby \pi_i$.
\end{proof}

The following proposition formalises the claim about path subsumption stated in \autoref{def:path} :
\begin{proposition}\label{prop:path-subsumption}
Let $\B_i=\bddtstupleind{i}$ be compatible BDDTSs for $i=1,2$, let $\B = \B_1 \bigtriangledown \B_2$ and let $\ini \in \TU(\emptyset)^V$. For all $i\in\{1,2\}$, $l_i\in L_i$ and $l\in\Proj[i]{l_i}$:
\[ LP_{\B}(\sigma,l) = \biguplus_{\pi\in LP_{\B_i}(\sigma,l_i)} LP_\B(\sigma,l\mid\pi) \enspace. \]
\end{proposition}
\begin{proof}
Fix $i\in\{1,2\}$, $l_i\in L_i$ and $l\in\Proj[i]{l_i}$. We prove the equality
by showing (1) the union on the right-hand side equals $LP_{\B}(\sigma,l)$
and (2) the union is disjoint.

\medskip\noindent
Suppose $\varpi \in \PCB{\discomp}{\sigma}{l}$
We define the \emph{projection} map
\[
\mathsf{proj}_i : LP_{\B}(\sigma,l) \longrightarrow LP_{\B_i}(\sigma,l_i)
\]
by sending a composition location-path $\varpi\in LP_{\B}(\sigma,l)$ to a component
location-path $\pi\in LP_{\B_i}(\sigma,l_i)$ such that
\(\pi \ssby \varpi\). in other words,: \[
\mathsf{proj}_i(\varpi) = \pi_i \quad \text{where } \pi_i \in LP_{\B_i}(\sigma, l_i) \text{ and } \pi_i \ssby \varpi.
\] Existence of such a $\pi$ for every
$\varpi\in LP_{\B}(\sigma,l)$ is assured by \autoref{lemma:path-subsumption}.
for every composition path there is a component path that it (component-)subsumes.
Hence $\mathsf{proj}_i$ is total (every $\varpi$ has at least one image).

Now for any $\varpi\in LP_{\B}(\sigma,l)$ let $\pi=\mathsf{proj}_i(\varpi)$.
By definition $\pi\in LP_{\B_i}(\sigma,l_i)$ and $\pi\ssby\varpi$, so
$\varpi\in LP_\B(\sigma,l\mid\pi)$. Thus every $\varpi\in LP_{\B}(\sigma,l)$
belongs to at least one set of the form $LP_\B(\sigma,l\mid\pi)$ with
$\pi\in LP_{\B_i}(\sigma,l_i)$. Therefore
\[
LP_{\B}(\sigma,l) \subseteq \bigcup_{\pi\in LP_{\B_i}(\sigma,l_i)} LP_\B(\sigma,l\mid\pi).
\]

Conversely, by definition each set $LP_\B(\sigma,l\mid\pi)$ is a subset of
$LP_{\B}(\sigma,l)$, so the reverse inclusion holds and we obtain equality
of the unions:
\[
LP_{\B}(\sigma,l) = \bigcup_{\pi\in LP_{\B_i}(\sigma,l_i)} LP_\B(\sigma,l\mid\pi).
\]
As $\B_1$, $\B_2$ and $\discomp$ are deterministic, this is a disjoint union.
\end{proof}
The following lemma relates the path condition of a path in one of the components of disjunction composition to path conditions of related paths --- namely, the subsuming ones --- in the composed BDDTS.

\begin{lemma} \label{lemma:single-path-eq}
Let $\B_i=\bddtstupleind{i}$ be compatible saturated BDDTSs for $i=1,2$, and let $\B = \B_1 \bigtriangledown \B_2=\bddtstuple$.  Let $\ini \in \TU(\emptyset)^V$, let $i\in \{1,2\}$ and $l_i\in L_i$ and $\sigma \in \IAU^*$, and consider a path $\pi_i\in LP_{\B_i}(\sigma,l_i)$; then
\begin{equation}
\eta^\ini_{\pi_i} 
\;\equiv\;
\bigvee_{l \in \Proj[i]{l_i} }\;\bigvee_{\pi \in LP_{\B}(\sigma, l \mid \pi_i)} \eta_\pi^\ini \enspace.
\end{equation}
\end{lemma}

\begin{proof}
By induction on length of $\sigma$. We give the proof for the case of $i=1$; the other case ($i=2$) is symmetrical.

\medskip\noindent
\textbf{Base case.} $|\sigma|=0$; then $l_1=\il_1$ and $\pi_1=\epsilon$; then $LP_\B(\sigma,l_1\mid\pi_1)=\{\epsilon\}$ and $\Proj[1]{il_1}=\{(il_1,il_2)\}$. Since $\eta^\ini_{\epsilon}=\True$ for any $\ini$, both sides of the equation are equivalent to $\True$.

\medskip\noindent
\textbf{Induction Hypothesis.} If $|\sigma| =k$ then for all $l_1\in L_1$ and $\pi_1\in LP_{\B_1}(\sigma,l_1)$
\[
\eta_{\pi_1}^{\ini} 
\;\equiv\;
\bigvee_{l \in \Proj[1]{l_1} }\;\bigvee_{\pi \in LP_{\B}(\sigma, l \mid \pi_1)} \eta_\pi^\ini \enspace.
\]
\textbf{Induction Step.} Assume $|\sigma|=k$ and consider 
$\sigma'=\sigma\alpha$ and $\pi_1'\in LP_{\B_1}(\sigma',l_1')$ for some $l_1'\in L_1$. It follows that $\pi_1'=\pi_1\cdot (\alpha,\phi_1,a_1)$ for some $\pi_1\in LP_{\B_1}(\sigma,l_1)$ (where $l_1$ is uniquely determined due to determinism), and hence
\begin{align*} 
\eta_{\pi_1'}^{\ini} 
 & \equiv 
   \eta_{\pi_1}^{\ini}\upshift \wedge \phi_1[a^\ini_{\pi_1}] \\
 \textit{(Ind. Hyp.)}\enspace
 & \equiv
   \left(\bigvee_{l \in \Proj[1]{l_1} }\;\bigvee_{\pi \in LP_{\B}(\sigma, l \mid \pi_1)} \eta_\pi^\ini \right)\upshift
       \wedge \phi_1[a^\ini_{\pi_1}] \\
 & \equiv
   \bigvee_{l \in \Proj[1]{l_1} }\left(\bigvee_{\pi \in LP_{\B}(\sigma, l \mid \pi_1)} \eta_\pi^\ini \upshift
       \wedge \phi_1[a^\ini_{\pi_1}]\right)
\end{align*}
As a final step to the proof obligation, it suffices to establish that, for any $l \in \Proj[1]{l_1}$,
\begin{equation}\label{eq:eta-step1}
\bigvee_{\pi \in LP_{\B}(\sigma, l \mid \pi_1)} \eta_\pi^\ini \upshift
       \wedge \phi_1[a^\ini_{\pi_1}\upshift]
\equiv
\bigvee_{\pi'\in LP_\B(\sigma',l'\mid\pi_1')} \eta^\ini_{\pi'} \enspace.
\end{equation}
To prove this, we use
\begin{align*}
LP_{\B}(\sigma',l'\mid\pi_1') = \{ \pi\cdot (\alpha,\phi,a) \mid &\exists l:\pi\in LP_{\B}(\sigma,l\mid\pi_1),\\
&l\trans{\alpha,\phi,a}, \phi_1\Leftarrow \phi, a_1\ssby a \}
\end{align*}
where $l$ is in fact uniquely determined by $\pi$; it follows that
\begin{align*}
&\bigvee_{\pi'\in LP_\B(\sigma',l'\mid\pi_1')} \eta^\ini_{\pi'}\\
& \equiv
  \bigvee_{\pi\in LP_\B(\sigma,l\mid\pi_1)}
  \;\bigvee_{l\trans{\alpha,\phi,a}, \phi_1\Leftarrow \phi, a_1\ssby a} \eta^\ini_{\pi\cdot (\alpha,\phi,a)} \\ 
& \equiv
  \bigvee_{\pi\in LP_\B(\sigma,l\mid\pi_1)}
  \;\bigvee_{l\trans{\alpha,\phi,a}, \phi_1\Leftarrow \phi, a_1\ssby a} (\eta^\ini_{\pi}\upshift \wedge \phi[a^\ini_\pi\upshift]) \\
& \equiv
  \bigvee_{\pi\in LP_\B(\sigma,l\mid\pi_1)}
  \eta^\ini_{\pi}\upshift \wedge \left(\bigvee_{l\trans{\alpha,\phi,a}, \phi_1\Leftarrow \phi, a_1\ssby a}  \phi[a^\ini_\pi\upshift]\right) \enspace.
\end{align*}
To close the proof of \eqref{eq:eta-step1}, we proceed by a case distinction on $l$. We know $l=(l_1,l_2)$ for some $l_2$. Either $\B_2$ has a transition $l_2\trans{\alpha,\phi_2,a_2}$ (for some $\phi_2$ and $a_2$), or there is no such transition.
\begin{itemize}
\item Assume $l_2\trans{\alpha,\phi_2,a_2}$ (implying $a_1\compat a_2$). Due to the fact that $\B_2$ is saturated, we have $\bigvee\Phi\equiv\True$ for $\Phi=\{\phi_2\mid l_2\trans{\alpha,\phi_2,a_2} \}$. Due to Rule~1 of \autoref{def:disjunction}, there is a one-to-one correspondence between (on the one hand) $\B$-transitions $l\trans{\alpha,\phi,a}$ such that $\phi_1\Leftarrow \phi$ and (on the other) $\B_2$-transitions $l_2\trans{\alpha,\phi_2,a_2}$, and that for each such correspondencs $\phi=\phi_1\wedge \phi_2$ and $a=a_1\sqcup a_2$. It follows that
\begin{align*}
\bigvee_{l\trans{\alpha,\phi,a}, \phi_1\Leftarrow \phi, a_1\ssby a}  &\phi[a^\ini_\pi\upshift]\\
   & \equiv
     \bigvee_{l_2\trans{\alpha,\phi_2,a_2}} (\phi_1\wedge \phi_2)[a^\ini_\pi\upshift] \\
   & \equiv
     \phi_1[a^\ini_\pi\upshift] \wedge \left(\bigvee_{l_2\trans{\alpha,\phi_2,a_2}} \phi_2\right)[a^\ini_\pi\upshift] \\
   & \equiv
     \phi_1[a^\ini_\pi\upshift] \wedge \True[a^\ini_\pi\upshift] \\
   & \equiv
     \phi_1[a^\ini_{\pi_1}\upshift]
\end{align*}

\item Assume $l_2$ has no outgoing $\alpha$-transition. Due to Rule~2 of  \autoref{def:disjunction}, it follows that there is exactly one $\B$-transition $l\trans{\alpha,\phi,a}$ for which $\phi_1\Leftarrow \phi$, and in that case $\phi=\phi_1$ and $a=a_1$. Hence
\[
\bigvee_{l\trans{\alpha,\phi,a}, \phi_1\Leftarrow \phi}  \phi[a^\ini_\pi\upshift]
   \;\equiv\;
     \phi_1[a^\ini_\pi\upshift] 
   \;\equiv\;
     \phi_1[a^\ini_{\pi_1}\upshift] \enspace.
\]
\end{itemize}
\end{proof}
Using the previous lemma, we can now relate the combined path conditions of all $\sigma$-generating paths up to a certain location of a component BDDTS to those of the $\sigma$-related paths of the injected locations of the composed BDDTS.

\begin{lemma} \label{lemma:equality-of-disjunctions}
Let $\B_i = \bddtstupleind{i}$ be compatible saturated BDDTSs for $i=1,2$ and let $\B = \B_1 \bigtriangledown \B_2=$ $ \bddtstuple$. Let $\ini \in \TU(\emptyset)^V$, and $l_i\in L_i$.
For any $\sigma \in \IAU^*$ it holds that:
    \[\bigvee_{\pi\in \PCB{\B_i}{\sigma}{l_i}} \etapi \equiv \bigvee_{l \in \Proj[i]{l_i}} \; \bigvee_{\pi \in \PCB{\B}{\sigma}{l}} \etapi \]
\end{lemma}

\begin{proof}
Using \autoref{lemma:single-path-eq} and \autoref{prop:path-subsumption} we have
\begin{align*}
\bigvee_{\pi_i\in \PCB{\B_i}{\sigma}{l_i}} \eta_{\pi_i}^{\ini}
 & \equiv\;
\bigvee_{\pi_i\in \PCB{\B_i}{\sigma}{l_i}}
  \;\bigvee_{l \in \Proj[i]{l_i}} 
  \;\bigvee_{\pi \in LP_{\B}(\sigma, l \mid \pi_i)} \eta_\pi^\ini \\
 & \equiv\;
\bigvee_{l \in \Proj[i]{l_i}}
  \;\bigvee_{\pi_i\in \PCB{\B_i}{\sigma}{l_i}}
  \;\bigvee_{\pi \in LP_{\B}(\sigma, l \mid \pi_i)} \eta_\pi^\ini \\
 & \equiv\;
\bigvee_{l \in \Proj[i]{l_i}}
  \;\bigvee_{\pi \in LP_{\B}(\sigma, l)} \eta_\pi^\ini \enspace.
\end{align*}
\end{proof}
The next step is to lift the previous result to the level of execution conditions.
\begin{lemma}[Execution conditions for disjunction composition]\label{lemma:EC-discopm-equality}
Let $\B_i = \bddtstupleind{i}$ for $i=1,2$ be compatible saturated BDDTSs and let $\B = \B_1 \bigtriangledown \B_2= \bddtstuple$. Let $\ini \in \TU(\emptyset)^V$.
For any $\sigma \in \IAU(G)^*$ it holds that:
\[ \EC_{\B_1,\ini}(\sigma) \vee \EC_{\B_2,\ini}(\sigma) \equiv  \EC_{\B,\ini}(\sigma) \enspace.
\]
\end{lemma}

\begin{proof}
We make a case distinction based on the length of $\sigma$. 

\noindent
\textbf{Case 1: $|\sigma|=0$}. Recall that the path condition of the empty sequence is $\True$ by definition.
\begin{align*}
\lefteqn{\EC_{\B_1,\ini}(\sigma) \vee \EC_{\B_2,\ini}(\sigma)}\qquad \\
 & \equiv (\IG_1 \land \eta^{\ini}_{\epsilon})
   \;\lor\; (\IG_2 \land\eta^{\ini}_{\epsilon})  \\
 & \equiv 
   \IG_1 \lor \IG_2\\
& \equiv
\EC_{\B,\ini}(\sigma)
\end{align*}
For the following two cases, we use the fact that (for any BDDTS $\B'$) there exists a path $\pi$ with $\sigma_\pi=\sigma$ if and only if there exists a location $l\in L_{\B'}$ such that $\pi\in LP_{\B'}(\sigma,l)$. Recall that $\B_1$ and $\B_2$ are saturated, and hence so is $\B$ by \autoref{prop:satcomp}.

\medskip\noindent
\textbf{Case 2: $|\sigma|=1$}. For the $\sigma=\sigma_\pi$ of length 1 the path condition is equal to the switch guard of outgoing switches from initial location:
\begin{itemize}
    \item For any $\il_{\B_1}\trans{\alpha,\phi_1,a_1}_{\B_1} l_1$, $\eta^{\ini}_{(\alpha,\phi_1,a_1)}= \phi_1$  and 
by condition 1 of \autoref{def:saturated} $\displaystyle\bigvee_{\{t \in \arrow_1 \mid \sl_t=\il_{\B_1}\}} \phi_t \equiv \True$, similarly
    \item For any $\il_{\B_2}\trans{\alpha,\phi_2,a_2}_{\B_2} l_2$, $\eta^{\ini}_{(\alpha,\phi_2,a_2)}= \phi_2$  and 
by condition 1 of \autoref{def:saturated} $\displaystyle\bigvee_{\{t \in \arrow_2 \mid \sl_t=\il_{\B_2}\}} \phi_t \equiv \True$, 
 \item For any $\il_{\B}\trans{\alpha,\phi,a}_{\B} l$, $\eta^{\ini}_{(\alpha,\phi,a)}= \phi$  and 
by condition 1 of \autoref{def:saturated} $\displaystyle\bigvee_{\{t \in \arrow \mid \sl_t=\il_{\B}\}} \phi_t \equiv \True$
\end{itemize}
\begin{align*}
\lefteqn{\EC_{\B_1,\ini}(\sigma) \vee \EC_{\B_2,\ini}(\sigma)}\qquad \\
 & \equiv \left(\IG_1 \land \left(\bigvee_{l_1\in L_1} \;\bigvee_{\pi\in \PCB{\B_1}{\sigma}{l_1}} \etapi \right)\right)
   \;\lor\;\\ 
   &\qquad \left(\IG_2 \land
   \left(\bigvee_{l_2\in L_2} \;\bigvee_{\pi\in \PCB{\B_2}{\sigma}{l_2}} \etapi \right)\right) \\
 & \equiv 
   (\IG_1  \land \True ) \;\lor\;
   ( \IG_2 \land \True )\\
    & \equiv 
    (\IG_1 \lor \IG_2)
    \;\land\; \left(\bigvee_{l \in L}
    \;\bigvee_{\pi \in \PCB{\B}{\sigma}{l}} \etapi \right)\\
    & \equiv 
    (\IG_1 \lor \IG_2)
    \;\land\; \True\\
     & \equiv \EC_{\B,\ini}(\sigma)
\end{align*}
\noindent
\textbf{Case 3: $\sigma > 1$}. For any path $\pi$ with $\sigma_\pi=\sigma$, if $(\alpha,\phi,a)$ is the first element of $\pi$ then $\il \trans{\alpha, \phi, a} l$ for a non-sink location $l$; hence $\phi \Rightarrow \IG$ according to Condition~3 of \autoref{def:saturated}. It follows that also $\etapi \Rightarrow \IG$, which means we can remove the $\IG$-conjunct from the definition of $\EC$ and use \autoref{lemma:equality-of-disjunctions}:
\begin{align*}
\lefteqn{\EC_{\B_1,\ini}(\sigma) \vee \EC_{\B_2,\ini}(\sigma)}\qquad \\
 & \equiv \left(\bigvee_{l_1\in L_1} \;\bigvee_{\pi\in \PCB{\B_1}{\sigma}{l_1}} \etapi \right)
   \;\lor\;\\ 
   &\qquad \left(\bigvee_{l_2\in L_2} \;\bigvee_{\pi\in \PCB{\B_2}{\sigma}{l_2}} \etapi \right) \\
 & \equiv
    \left( \bigvee_{l_1\in L_1} \;\bigvee_{l \in \Proj[1]{l_1}}
   \;\bigvee_{\pi \in \PCB{\B}{\sigma}{l}} \etapi \right) 
   \;\lor\;\\
   &\qquad \left( \bigvee_{l_2\in L_2} \;\bigvee_{l \in \Proj[2]{l_2}}
   \;\bigvee_{\pi \in \PCB{\B}{\sigma}{l}} \etapi \right) \\
    & \equiv 
    \left(\bigvee_{l \in L}
   \;\bigvee_{\pi \in \PCB{\B}{\sigma}{l}} \etapi \right)\\
    & \equiv \EC_{\B,\ini}(\sigma) \enspace.
\end{align*}
\end{proof}

\begin{lemma} \label{lemma:single-conjunction}
Let $\B_i = \bddtstupleind{i}$ be compatible saturated BDDTSs for $i=1,2$ and let $\B = \B_1 \bigtriangledown \B_2=$ $ \bddtstuple$. Let $\ini \in \TU(\emptyset)^V$, $i\in\{1,2\}$ and $l_i\in L_i$ such that $\OG_i \definedon l_i$.
Let $\sigma \in \IAU(G)^*$ and  consider a path $\pi_i\in LP_{\B_i}(\sigma,l_i)$; then it holds that:
\begin{align*}
   &\left(\eta^{\ini}_{\pi_i} \Rightarrow \OG_{\B_i}(l_i)[a^{\ini}_{\pi_i}] \right)\\
   &\qquad \equiv  \bigwedge_{ l \in \Proj[i]{l_i}}\bigwedge_{\pi \in \PCB{\B}{\sigma}{l \mid \pi_i}} \left( \etapi \Rightarrow \OG_{\B_i}(l_i)[a^{\ini}_{\pi_i}]\right)
 \end{align*}
\end{lemma}
\begin{proof}
By applying \autoref{lemma:single-path-eq} on the LHS, we get 
\begin{align*}
   & \left(\eta^{\ini}_{\pi_i} \Rightarrow \OG_{\B_i}(l_i)[a^{\ini}_{\pi_i}] \right)\\
   &\qquad \equiv
\left(\bigvee_{l \in \Proj[i]{l_i} }\;\bigvee_{\pi \in LP_{\B}(\sigma, l \mid \pi_i)} \eta_\pi^\ini \right) \Rightarrow \OG_{\B_i}(l_i)[a^{\ini}_{\pi_i}]
\end{align*}
\noindent
By disjunction elimination we know $(A \lor B) \Rightarrow C \equiv (A\Rightarrow C) \land (B \Rightarrow C)$:
\begin{align*}
\lefteqn{\left(\bigvee_{l \in \Proj[i]{l_i} }\;\bigvee_{\pi \in LP_{\B}(\sigma, l \mid \pi_i)} \eta_\pi^\ini \right) \Rightarrow \OG_{\B_i}(l_i)[a^{\ini}_{\pi_i}] }
\qquad \\
 & \equiv
\bigwedge_{l \in \Proj[i]{l_i} }\;\bigwedge_{\pi \in LP_{\B}(\sigma, l \mid \pi_i)} \left(\eta_\pi^\ini  \Rightarrow \OG_{\B_i}(l_i)[a^{\ini}_{\pi_i}] \right)
\end{align*}
which closes the proof.
\end{proof}

\begin{lemma}[Goal Implication for disjunction composition]\label{lemma:GI-discomp-equality}
Let $\B_i = \bddtstupleind{i}$ for $i=1,2$ be compatible saturated BDDTSs and let $\B = \B_1 \bigtriangledown \B_2=\bddtstuple$. Let $\ini \in \TU(\emptyset)^V$. For any $\sigma \in \IAU^*$ it holds that:
\[
\GI_{\B_1,\ini}(\sigma) \wedge \GI_{\B_2,\ini}(\sigma) \equiv \GI_{\B,\ini}(\sigma) \enspace.
\]
\end{lemma}
\begin{proof} We rewrite the proof obligation in terms of location paths sets $LP$:
     \[ \left( \bigwedge_{ \substack{l_1 \in L_1 \land \\ \OG_1 \definedon l_1}}\bigwedge_{\pi \in \PCB{\B_1}{\sigma}{l_1}} (\etapi \Rightarrow \OG_{\B_1}(l_1)[a^{\ini}_\pi])\right) \land \]
\[\left(\bigwedge_{\substack{ l_2 \in L_2 \land \\ \OG_2 \definedon l_2}} \bigwedge_{\pi \in \PCB{\B_2}{\sigma}{l_2}} (\etapi \Rightarrow \OG_{\B_2}(l_2)[a^{\ini}_\pi]) \right) \equiv\]
\[ \left( \bigwedge_{ \substack{l \in L \land \\ \OG \definedon L}}\bigwedge_{\pi \in \PCB{\B}{\sigma}{l}} (\etapi \Rightarrow \OG_{\B}(l)[a^{\ini}_\pi]) \right)\]

\noindent
\textbf{Left $\Rrightarrow$ Right}:
We choose an arbitrary $l$ and $\pi$ in $\discomp$.
There are 3 cases based on the structure of $l$:
\begin{enumerate}
    \item $l=(l_1,l_2)$ where $l_1 \in L_1$ and $l_2 \in L_2$.
    \begin{enumerate}
    \item $\OG_1 \definedon l_1$ and $\OG_2 \definedon l_2$, By definition of disjunction composition $\OG_{\B}(l)[a^{\ini}_{\pi}]= \OG_{\B_1}(l_1)[a^{\ini}_{\pi_1}] \land \OG_{\B_2}(l_2)[a^{\ini}_{\pi_2}])$, thus for this case we need to show for any arbitrary $\pi_1 \in \PCB{\B_1}{\sigma}{l_1} $ and $\pi_2 \in \PCB{\B_2}{\sigma}{l_2} $ such that $\pi_1 \ssby \pi$ and $\pi_2 \ssby \pi$:
    \begin{align*}\lefteqn{
    (\eta^{\ini}_{\pi_1} \Rightarrow\OG_{\B_1}(l_1)[a^{\ini}_{\pi_1}])\land (\eta^{\ini}_{\pi_2} \Rightarrow\OG_{\B_2}(l_2)[a^{\ini}_{\pi_2}])
    }\qquad \\
    & \Rrightarrow\quad
    \eta^{\ini}_{\pi} \Rightarrow( \OG_{\B_1}(l_1)[a^{\ini}_{\pi_1}] \land \OG_{\B_2}(l_2)[a^{\ini}_{\pi_2}])
    \end{align*}
    By \autoref{lemma:path-subsumption} we know $\exists \pi_i \in \PCB{\B_i}{\sigma}{l_i}$ such that $\pi_i \ssby \pi$, and as a result of \autoref{lemma:single-path-eq} $\etapi \Rightarrow \eta^{\ini}_{\pi_i}$ and applyin the rule $(a \Rightarrow b) \land (a \Rightarrow c) \equiv a \Rightarrow (b \land c)$, then we have:
    \begin{align*}\lefteqn{
    (\eta^{\ini}_{\pi} \Rightarrow  \OG_{\B_1}(l_1)[a^{\ini}_{\pi_1}]) \land (\eta^{\ini}_{\pi} \Rightarrow\OG_{\B_2}(l_2)[a^{\ini}_{\pi_2}])
    }\qquad\\
    & \equiv \quad
    \eta^{\ini}_{\pi} \Rightarrow( \OG_{\B_1}(l_1)[a^{\ini}_{\pi_1}] \land \OG_{\B_2}(l_2)[a^{\ini}_{\pi_2}])
    \end{align*}

    \item $\OG_1 \definedon l_1$ and $\OG_2 \undefinedon l_2$, By definition of disjunction $\OG_{\B}(l)[a^{\ini}_{\pi}]= \OG_{\B_1}(l_1)[a^{\ini}_{\pi_1}]$, and as $\OG$ is not defined for $l_2$ we only need to show for any arbitrary $\pi_1 \in \PCB{\B_1}{\sigma}{l_1} $  such that $\pi_1 \ssby \pi$. it holds that:
    \[ (\eta^{\ini}_{\pi_1} \Rightarrow\OG_{\B_1}(l_1)[a^{\ini}_{\pi_1}]) \Rrightarrow          (\eta^{\ini}_{\pi} \Rightarrow \OG_{\B_1}(l_1)[a^{\ini}_{\pi_1}])
    \]
    From \autoref{lemma:path-subsumption} we know $\exists \pi_i \in \PCB{\B_i}{\sigma}{l_i}$ such that $\pi_i \ssby \pi$, and as a result of \autoref{lemma:single-path-eq} $\etapi \Rightarrow \eta^{\ini}_{\pi_i}$, thus 
        $(\eta^{\ini}_{\pi} \Rightarrow \OG_{\B_1}(l_1)[a^{\ini}_{\pi_1}])$

    \item 
        $\OG_1 \undefinedon l_1$ and $\OG_2 \definedon l_2$, the proof is analogous to Case "1b", with $l_2$ in place of $l_1$ and $\pi_2$ in place of $\pi_1$.
    \end{enumerate}
    \item $l=(l_1, \bot)$ for $l_1 \in L_1$ and $l_2 \in \{\bot\}$:  the proof is analogous to Case 1(b)
    \item $l=(\bot, l_2)$  for $l_2 \in L_2$ and $l_1 \in \{\bot\}$: the proof is analogous to Case 1(c).
\end{enumerate}
\textbf{Right $\Rrightarrow$ Left}: We first choose an arbitrary $l_1 \in L_1$ and $l_2 \in L_2$ \\
and we rewrite the RHS based on the injection sets of $l_1$ and $l_2$:
\[ \left( \bigwedge_{ \substack{l \in L \land \\ \OG \definedon L}}\bigwedge_{\pi \in \PCB{\B}{\sigma}{l}} (\etapi \Rightarrow \OG_{\B}(l)[a^{\ini}_\pi]) \right) \Rrightarrow\]
\begin{align*}
   & \left(\bigwedge_{ l \in \Proj[1]{l_1} }\bigwedge_{\pi \in \PCB{\B}{\sigma}{l}} (\etapi \Rightarrow \OG_{\B_1}(l_1)[a^{\ini}_\pi])\right) \land\\
   &\qquad \left(\bigwedge_{ l \in \Proj[2]{l_2}}\bigwedge_{\pi \in \PCB{\B}{\sigma}{l}} (\etapi \Rightarrow \OG_{\B_2}(l_2)[a^{\ini}_\pi])\right) 
   \end{align*}
Below, we apply 2 facts: first, using the fact that $(A \land B) \Rightarrow (C \land D) \equiv (A \land B \Rightarrow C ) \land (A \land B \Rightarrow D ) $, we show it for the big conjunction for $l \in \Proj[1]{l_1}$:
\begin{align*}
   &  \bigwedge_{ \substack{l \in L \land \\ \OG \definedon L}}\bigwedge_{\pi \in \PCB{\B}{\sigma}{l}} (\etapi \Rightarrow \OG_{\B}(l)[a^{\ini}_\pi]) \\
   &\qquad \Rrightarrow \bigwedge_{ l \in \Proj[1]{l_1} }\bigwedge_{\pi \in \PCB{\B}{\sigma}{l}} (\etapi \Rightarrow \OG_{\B_1}(l_1)[a^{\ini}_\pi]) 
   \end{align*}
Second, Using \autoref{prop:path-subsumption}, we expand the inner conjunctions of the RHS, and using the fact that for all $\pi \in \PCB{\B_i}{\sigma}{ l \mid \pi_i}$, $a_{\pi_i} \ssby a_{\pi}$ and the compatibility of assignments in the disjunction, then for any $\pi_i \ssby\pi$, $\OG_{\B_i}(l_i)[a^{\ini}_{\pi}]\equiv \OG_{\B_i}(l_i) [a^{\ini}_{\pi_i}]$:
\[
\begin{aligned}
& \bigwedge_{ l \in \Proj[1]{l_1} }\bigwedge_{\pi \in \PCB{\B}{\sigma}{l}} (\etapi \Rightarrow \OG_{\B_1}(l_1)[a^{\ini}_\pi]) \\
& \equiv \underbrace{ 
    \bigwedge_{ l \in \Proj[1]{l_1} } 
    \bigwedge_{\pi_1 \in \PCB{\B_1}{\sigma}{l_1}} 
    \bigwedge_{\pi \in \PCB{\B_1}{\sigma}{l \mid \pi_1}}
    (\etapi \Rightarrow \OG_{\B_1}(l_1)[a^{\ini}_{\pi_1}]) 
}_{\text{By \autoref{prop:path-subsumption}, $a_{\pi_i}\ssby a_\pi$, and compatibility of assignments}} 
\\[1em]
& \equiv \underbrace{ 
    \bigwedge_{\pi_1 \in \PCB{\B_1}{\sigma}{l_1}} 
    \bigwedge_{ l \in \Proj[1]{l_1}}
    \bigwedge_{\pi \in \PCB{\B_1}{\sigma}{l \mid \pi_1}}
    (\etapi \Rightarrow \OG_{\B_1}(l_1)[a^{\ini}_{\pi_1}]) 
}_{\text{Swap leftmost and middle conjunctions since $\PCB{\B_1}{\sigma}{l_1}$ is independent of $l$}} 
\\[1em]
& \equiv \underbrace{ 
    \bigwedge_{\pi_1 \in \PCB{\B_1}{\sigma}{l_1}} 
    (\eta^{\ini}_{\pi_1} \Rightarrow \OG_{\B_1}(l_1)[a^{\ini}_{\pi_1}]) 
}_{\text{By \autoref{lemma:single-conjunction}}} 
\end{aligned}
\]
until so far, we proved for any goal locations $l_1 \in L_1$:
\begin{align*}
   &\bigwedge_{ l \in \Proj[1]{l_1}}\bigwedge_{\pi \in \PCB{\B}{\sigma}{l}} (\etapi \Rightarrow \OG_{\B_1}(l_1)[a^{\ini}_\pi])\\
   &\qquad \equiv \bigwedge_{\pi_1 \in \PCB{\B_1}{\sigma}{l_1}} (\eta^{\ini}_{\pi_1} \Rightarrow \OG_{\B_1}(l_1)[a^{\ini}_{\pi_1}])
   \end{align*}
The same reasoning applies for goal locations $l_2 \in L_2$ resulting in :
\begin{align*}
   &\bigwedge_{ l \in \Proj[2]{l_2}}\bigwedge_{\pi \in \PCB{\B}{\sigma}{l}} (\etapi \Rightarrow \OG_{\B_2}(l_2)[a^{\ini}_\pi]) 
   \\ &\qquad \equiv
\bigwedge_{\pi_2 \in \PCB{\B_2}{\sigma}{l_2}} (\eta^{\ini}_{\pi_2} \Rightarrow \OG_{\B_2}(l_2)[a^{\ini}_{\pi_2}])
\end{align*}
The two statements above covers the cases where only $l_1$ \emph{or} $l_2$ are goal locations, 
combining the two cases above where both $\OG_1 \definedon l_1$ and $\OG_2 \definedon l_2$ results in:
\begin{align*}
   & \left( \bigwedge_{ \substack{l \in L \land \\ \OG \definedon L}}\bigwedge_{\pi \in \PCB{\B}{\sigma}{l}} (\etapi \Rightarrow \OG_{\B}(l)[a^{\ini}_\pi]) \right) \Rrightarrow\\
   &\quad \left(\bigwedge_{\pi_1 \in \PCB{\B_1}{\sigma}{l_1}} (\eta^{\ini}_{\pi_1} \Rightarrow \OG_{\B_1}(l_1)[a^{\ini}_{\pi_1}])\right)\\
   &\qquad \land \left(\bigwedge_{\pi_2 \in \PCB{\B_2}{\sigma}{l_2}} (\eta^{\ini}_{\pi_2} \Rightarrow \OG_{\B_2}(l_2)[a^{\ini}_{\pi_2}]) \right)
   \end{align*}
As we chose arbitrary locations $l_1$ and $l_2$ this can be applied to any $l_1 \in L_1$ and $l_2 \in L_2$, closing the proof. 
\end{proof}

\noindent
\autoref{lemma:EC-discopm-equality} and \autoref{lemma:GI-discomp-equality} both hold for any initialisation $\ini$, regardless of whether this satisfies the input guards of the BDDTSs involved. If an input guard is \emph{not} satisfied for a given $\ini$, the execution condition and goal implication collapse to $\False$, respectively $\True$, for every non-empty interaction sequence. To formulate this precisely, we use the notation $\phi_a$ for an assignment $a\in \TU(\emptyset)^\LV_\bot$ to stand for the formula $\bigwedge_{a\definedon x} x=a(x)$.
\begin{lemma}\label{lem:IG-violated}
Let $\B=\bddtstuple$ be a saturated BDDTS and $\ini\in \TU(\emptyset)^V$ such that $\sem\ini\not\models\IG$. For all $\sigma\in \IAU(G)^*$, $\phi_\ini\wedge \EC_{\B,\ini}(\sigma)\equiv\False$ and $\phi_\ini\Rightarrow \GI_{\B,\ini}(\sigma)$.
\end{lemma}

\begin{proof}
\leavevmode
\begin{enumerate}
  \item 
  By definition $\phi_\ini$ is the formula that fixes the initial valuation to $\ini$. Hence $\sem\ini\not\models\IG$ is equivalent to the semantic consequence $\phi_\ini\models\neg\IG$, i.e.\ $\phi_\ini\Rightarrow\neg\IG$. (Informally: starting from $\phi_\ini$ the input guard $\IG$ is false.)

  \item
  By the definition of the execution condition $\EC_{\B,\ini}(\sigma)$ for interaction sequences, any non-empty sequence $\sigma$ can only be executable if the input guard holds at the start. Concretely
  \[
  \EC_{\B,\ini}(\sigma)\;\Longrightarrow\;\IG
  \]
  for all $|\sigma|\ge 1$.

  \item
  Combine steps 1 and 2: from $\phi_\ini\Rightarrow\neg\IG$ and $\EC_{\B,\ini}(\sigma)\Rightarrow\IG$ we get
  \[
  \phi_\ini\land\EC_{\B,\ini}(\sigma)\;\Rightarrow\;\neg\IG\land\IG,
  \]
  which is unsatisfiable. Hence $\phi_\ini\land\EC_{\B,\ini}(\sigma)\equiv\False$ for every non-empty $\sigma$.
  \item For Goal implication Suppose $\il\trans\pi l$ for some goal location $l$. 
By Condition~3 of \autoref{def:saturated}, the first switch leads to a location $l'$ 
that is either equal to $l$ itself or to a non-sink state. 
In either case we obtain
\[
  \eta_\pi^\ini \;\Longrightarrow\; \IG .
\]
Since by assumption $\phi_\ini \Rightarrow \neg \IG$, it follows that
\[
  \phi_\ini \;\Longrightarrow\; \neg \eta_\pi^\ini .
\]
Now recall the definition of $\GI_{\B,\ini}(\sigma)$:
\begin{align*}
   &
  \GI_{\B,\ini}(\sigma)
  \;=\;\\
  &\quad \bigwedge\Big\{\, \eta_\pi^\ini \Rightarrow \OG(l)[a_\pi^\ini]
  \;\Big|\;
  \il\trans\pi l,\; \OG\downarrow l,\; \sigma_\pi = \sigma \,\Big\}.
\end{align*}
For each conjunct, under $\phi_\ini$ the antecedent $\eta_\pi^\ini$ is false, 
hence the implication is true (vacuous truth). 
Therefore every conjunct holds under $\phi_\ini$, so $\phi_\ini \Rightarrow (\eta_\pi^\ini\Rightarrow \OG(l))$ and hence
\[
  \phi_\ini \;\implies\; \GI_{\B,\ini}(\sigma).
\]
  
\end{enumerate}
This completes the proof for every non-empty sequence $\sigma$.
\end{proof}

\disjunctiontraceequive*

\begin{proof}
Let $\BB_1=\{\B_1,\B_2\}$ (the left hand side) and $\BB_2=\{\B_1\bigtriangledown\B_2\}$ (the right hand side) and let $\ini\in \TU(\emptyset)^V$. If $\sem\ini\not\models \IG$ then also $\sem\ini\not\models\IG_i$ for $i=1,2$, hence $\BB_1\restr\ini= \BB_2\restr\ini=\emptyset$ and the conditions for $\simeq$ are trivially fulfilled.

Otherwise, $\sem\ini\models \IG$ and hence $\sem\ini\models\IG_i$ for either $i=1$ or $i=2$ or both. From \autoref{def:testing-equivalence}, $\BB_1\simeq\BB_2$ holds if the following are satisfied for all $\sigma \in \IAU^*$ and  $\ini \in \TU(\emptyset)^V$: 
\begin{enumerate}
\item $\bigvee_{\B\in\{\B_1,\B_2\}\restr\ini} \EC_{\B,\ini}(\sigma)
  \equiv  \EC_{\B,\ini}(\sigma)$: This holds by \autoref{lemma:EC-discopm-equality} in combination with \autoref{lem:IG-violated}.
\item $\bigwedge_{\B\in\{\B_1,\B_2\}\restr\ini} \GI_{\B,\ini}(\sigma)
  \equiv \GI_{\B,\ini}(\sigma)$: This holds by \autoref{lemma:GI-discomp-equality} in combination with \autoref{lem:IG-violated}.
\end{enumerate}
\end{proof}

\section{Proofs of \autoref{sec:testcases} }
\propsaturated*
\begin{proof}\emph{(Sketch)}
There are two cases to consider.
\begin{itemize} 
\item[$\Rightarrow$] Assume $\TC(\B,\ini)\trans{\omega} q \trans u q'\in F$ (implying $q'=q_\f$). It follows that there is a path $l_0\trans\pi l$ in $\B$ with $l\in L^\Cl$ and $|\pi|=|\omega|$, giving rise to a path $l_0\trans{\pi'} l$ in $\ISS\B\ini$ where $|\pi'|=|\omega|$, such that $q=(l,\vartheta)$ for some $\vartheta$ and the values in $\omega$ satisfy the corresponding successive switch guards of $\pi'$. Since saturation only adds switches, $\pi'$ also exists in $\ISS{\sat\B}\ini$, except that $\IG$ is added to its initial switch; however, this does not affect the outcome since $\IG$ is known to be satisfied by $\ini$. We may conclude that $\TC(\sat\B,\ini)\trans\omega q$ as well.

Moreover, $q \trans u q_\f$ implies that $g_u\in A_o$ and $q\ntrans u$ in $\sem{\ISS\B\ini}$. This means that, for all $l\trans{\alpha,\phi',a}$ in $\ISS\B\ini$ with $g_\alpha=g_u$, $\vartheta_n\uplus \vartheta_u$ does not satisfy $\phi'$. Any transition $l\trans{\alpha,\phi,a}$ added in $\sat\B$ by saturation is bound to lead to $\satbot$ of which $\sat\OG(\satbot)=\False$, hence in $\ISS{\sat\B}\ini$ the corresponding switch guard is conjoined with $\False$, ensuring that it is unsatisfiable. It follows that $q\ntrans u$ in $\sem{\ISS{\sat\B}\ini}$; hence $q\trans u q_\f$ in $\TC(\sat\B,\ini)$.

\item[$\Leftarrow$] Assume $\TC(\sat\B,\ini)\trans\omega q \trans u q'\in F$ (implying $q'=q_\f$). It follows that there is a path $l_0\trans\pi l$ in $\sat\B$ with $l\in L^\Cl$ and $|\pi|=|\omega|$, giving rise to a path $l_0\trans{\pi'} l$ in $\ISS{\sat\B}\ini$ where $|\pi'|=|\omega|$, such that $q=(l,\vartheta)$ for some $\vartheta$ and the values in $\omega$ satisfy the corresponding successive switch guards of $\pi'$. $q\trans u q_\f$ only depends on the \emph{absence} of certain switches from $l$ in $\sat\B$, which therefore certainly do not exist in $\B$ either (saturation only adds switches), hence $q\trans u q_\f$ also in $\TC(\B,\ini)$. We therefore only have to argue that all switches in $\pi$ already existed in $\B$ and none of them were added or modified by saturation. For $1\leq i<|\pi|$ this follows from the fact that the target state is not a sink state (all saturating switches lead to sink states), also noting that the addition of $\IG$ to the initial switch guard does not change its satisfaction because $\IG$ is known to be satisfied by $\ini$; whereas for $i=|\pi|$ it follows from the fact that $l$ is closed (all saturating switches lead to open states).
\end{itemize}
\end{proof}
The following proposition states how, for a given STS $\S$, the $\vu_\omega$ on the one hand and the path conditions on the other can be used to precisely characterise the paths in $\sem\S$ in terms of those in $\S$. This is essentially the justification of the path conditions.
Below, $a^\S_\pi$ and $\eta^\S_\pi$ stand for the path assignment and path condition for $\pi$, computed according to the initialization function and switch guards of the STS $\S$.
\begin{restatable}{proposition}{propstslts}\label{prop:eta-vs-vu}
Let $\S$ be an arbitrary STS and $\L=\sem{\S}$ the corresponding LTS; then $(\il,\sem\ini)\trans\omega q$ if and only if $\il\trans\pi l$ for some $\pi$ and $l$ such that $\sigma_\pi=\sigma_\omega$ and $\vu_\omega\models \eta^\S_\pi$; moreover, in that case $q=(l,\sem{a^\S_\pi}_{\vu_\omega})$. Finally, if $\S$ is deterministic, then $\pi$ and $l$ are unique.
\end{restatable}
\begin{proof}
By induction on the length of $\omega$ and $\pi$. For $\omega=\epsilon$ and $\pi=\epsilon$ the property is immediate (take $l=\il$ and $\omega=\pi=\epsilon$; the conditions of the proof obligation are trivially satisfied by $\eta^\S_\pi=\True$ and $a^\S_\pi=\ini$). Now assume the proposition holds for a given $\pi$, resp.\ $\omega$.
\begin{description}
\item[If.] Assume $\il\trans\pi l\trans{\alpha,\phi,a} l'$ and let $u\in\GUU$ be such that $g_u=g_\alpha$ and $\vu_{\omega u}\models \eta^\S_{\pi\cdot(\alpha,\phi,a)}$. Due to $\vu_{\omega u}=\upshift\vu_\omega\sqcup \vu_u$ and $\eta^\S_{\pi\cdot(\alpha,\phi,a)}=\eta^\S_\pi\upshift \wedge \phi[a^\S_\pi\upshift]$ we have 
\begin{itemize}
\item $\upshift\vu_\omega\models \eta^\S_\pi\upshift$, which holds if and only if $\vu_\omega\models \eta^\S_\pi$. By the induction hypothesis, it follows that $(\il,\sem\ini)\trans\omega (l,\vu)$ for $\vu=\sem{a^\S_\pi}_{\vu\omega}$.
\item $\upshift\vu_\omega\sqcup \vu_u\models \phi[a^\S_\pi\upshift]$, which holds if and only if $\sem{a^\S_\pi\upshift}_{\upshift\vu_\omega}\sqcup \vu_u\models \phi$, which in turn, due to $\sem{a^\S_\pi\upshift}_{\upshift\vu_\omega}=\sem{a^\S_\pi}_{\vu_\omega}$, holds if and only if $\vu\sqcup \vu_u\models \phi$. According to the construction of $\L$ (\autoref{def:intrp}), we therefore have $(l,\vu) \trans{u} (l',\sem a_{\vu\sqcup \vu_u})$. The proof is then completed by the observation that
\begin{equation}\begin{array}{rl}\label{eq:obs}
\sem{a^\S_{\pi\cdot(\alpha,\phi,a)}}_{\vu_{\omega u}}
 & = \sem{a^\S_\pi\upshift\circ a}_{\upshift\vu_\omega\sqcup \vu_u} \\[\smallskipamount]
 & = \sem{a}_{\sem{a^\S\pi\upshift}_{\upshift\vu_\omega}\sqcup \vu_u} \\[\smallskipamount]
 & = \sem{a}_{\sem{a^\S\pi}_{\vu_\omega}\sqcup \vu_u} \\[\smallskipamount]
 & = \sem{a}_{\vu\sqcup \vu_u} \enspace.
\end{array}\end{equation}
\end{itemize}

\item[Only if.] Assume $(\il,\ini)\trans\omega q\trans u q'$. The induction hypothesis implies $\il\trans\pi l$ such that $\sigma_\pi=\sigma_\omega$, $\vu_\omega\models\eta^\S_\pi$ and $q=(l,\vu)$ for $\vu=\sem{a^\S_\pi}_{\vu_\omega}$; furthermore, there is a $l\trans{\alpha,\phi,a} l'$ such that $g_u=g_\alpha$, $\vu\sqcup \vu_u\models \phi$ and $q'=(l',\sem a_{\vu\sqcup \vu_u})$.

From $\vu_\omega\models\eta^\S_\pi$ and $\vu\sqcup \vu_u\models \phi$, using the same steps as in the ``if'' part (all of which were if-and-only-if) we can deduce that $\vu_{\omega u}\models \eta^\S_{\pi\cdot(\alpha,\phi,a)}$. Moreover, observation \eqref{eq:obs} now implies $q'=(l',\sem{a^\S_{\pi\cdot(\alpha,\phi,a)}}_{\vu_{\omega u}})$, closing the proof.
\end{description}
\end{proof}

The next sequence of lemmas leads up to the proof of \autoref{thm:sym-concrete-relation}. First we show that, if the valuation of a gate value sequence satisfies the execution condition of its underlying interaction sequence, then this automatically also holds for all prefixes.

\begin{lemma}\label{theta-prefix}
Let $\B$ be a saturated output-rich BDDTS and let $\ini \in \TU(\emptyset)^V$ such that $\sem\ini \models \IG_\B$. For all $\xi,\omega\in \GUU^*$ such that $\xi\preceq \omega$, if $\vu_\omega \models\EC_{\B,\ini}(\sigma_\omega)$ then $\vu_\xi \models \EC_{\B,\ini}(\sigma_\xi)$.
\end{lemma}

\begin{proof}
Using \autoref{def:goal-implication}, the definition of $\eta$ and logical reasoning:
\begin{align*} \EC_{\B,\ini}&(\sigma_{\omega u})\\
 & = \IG_{\B} \land \textstyle\bigvee \{ \eta_{\pi\cdot(\alpha,\phi,a)}^\ini \mid \il\trans{\pi\cdot(\alpha,\phi,a)}, \sigma_{\pi\cdot(\alpha,\phi,a)} =\sigma_{\omega u} \} \\
 & = \IG_{\B} \land \textstyle\bigvee \{ \eta_\pi^\ini\upshift \wedge \phi[a_\pi^\ini\upshift] \mid \il\trans{\pi}\trans{\alpha,\phi,a}, \sigma_\pi=\sigma_\omega \} \\
 & \Models \IG_{\B} \land \textstyle\bigvee \{ \eta_\pi^\ini\upshift \mid \il\trans{\pi}, \sigma_\pi=\sigma_\omega \} \\
 & = \EC_{\B,\ini}(\sigma_\omega)\upshift \enspace.
\end{align*}
Furthermore, $\vu_\omega$ satisfies the recursive equation $\vu_{\omega u} = \vu_\omega\upshift \sqcup \vu_u$. It follows that $\vu_{\omega u}\models \EC_{\B,\ini}(\sigma_{\omega u})$ implies $\vu_\omega\models \EC_{\B,\ini}(\sigma_\omega)$; by transitivity, this implies the statement of the lemma.
\end{proof}

Next, we show the relation between (satisfying) the path conditions of a BDDTS and those of the derived STS (in which the output guards have been shifted to the switch guards).

\begin{lemma}\label{lem:eta-B-to-S}
Let $\B=\bddtstuple$ be a saturated output-rich BDDTS and let $\S=\ISS{\B}{\ini}$ be its derived STS.  
Consider any gate value sequence $\omega \in \GU\UU^*$ Then the following two statements are equivalent
\begin{enumerate}
\item There is a path $\pi^\S \in \Lambda_\S^*$ in $\S$ such that
  \[
    \sigma_\omega = \sigma_{\pi^\S} 
    \quad\text{and}\quad 
    \vu_\omega \models \eta^\S_{\pi^\S}.
  \]

\item There is a path $\pi^\B \in \Lambda_B^*$  in $\B$ such that
  \[
    \sigma_\omega = \sigma_{\pi^\B} 
    \quad\text{and}\quad 
    \vu_\omega \models \eta^\ini_{\pi^\B},
  \]
  and moreover, for every prefix $\xi \preceq \omega$ there exists a path $\varpi \preceq \pi^\B$ with $\sigma_\xi = \sigma_\varpi$, such that for all locations $l$ where $\B \trans{\varpi} l$ and $\OG \definedon l$, the prefix valuation satisfies
  \[
    \hat\vu_\xi \models \OG(l).
  \]
\end{enumerate}
\end{lemma}

\begin{proof}
\textbf{From 1 to~2}\\
\textbf{Assumption:}\[\sigma_\omega = \sigma_{\pi^\S} 
    \quad\text{and}\quad 
    \vu_\omega \models \eta^\S_{\pi^\S}. \]
\\
\textbf{To prove:} There is a $\pi^\B$ such that
\begin{itemize}
\item $\sigma_\omega = \sigma_{\pi^\B}$ and $\vu_\omega \models \eta^\ini_{\pi^\B}$;
\item For all $\B\trans\varpi l$ with $\varpi\preceq \pi^\B$ and $\OG\definedon l$, if $\xi\preceq \omega$ such that $\sigma_\xi=\sigma_\varpi$ then $\hat\vu_\xi \models \OG(l)$.
\end{itemize}
We know from definition \autoref{def:bddtstosts}, the set of switches in $\ISS{\B}{\ini}$ is:
    \begin{align*}
        \arrow_c =\ &  \{(sl_t,\alpha_t, \phi_t, a_t , tl_t) \mid t \in \arrow, \OG\undefinedon \tl_t \} \\
 \cup\ & \{ (sl_t, \alpha_t, \phi_t \wedge \OG[\mapcv^{g_t}], a_t, tl_t) \mid  t \in  \arrow,\ \OG \definedon tl_t \}
    \end{align*}
      meaning that first the STS and BDDTS share the same sequence of interactions and second,the switch labels of a path in STS, is either equal to ones in the BDDTS (for switches ending in a non-goal location) or the switch guard is different with regard to the partial renaming function $\rho$.

We prove the lemma by the induction on length of $\omega$ which by definition has the same length as $\pi^\S$.

\noindent
\textbf{Base case} is trivially true
for $\epsilon$, by definition the path conditions are evaluated to $\True$ and as there is no prefix that leads to a goal state and the fact that we only have output rich systems the OG cannot be placed in the initial location and the base case holds.

\noindent
\textbf{Inductive Hypothesis} for the path with length $n$ we assume that when $ \sigma_\omega = \sigma_{\pi^\S} 
    \quad\text{and}\quad 
    \vu_\omega \models \eta^\S_{\pi^\S}$ then $\sigma_\omega = \sigma_{\pi^\B} 
    \quad\text{and}\quad 
    \vu_\omega \models \eta^\ini_{\pi^\B},$ and  for every prefix $\xi \preceq \omega$ with length $|\xi| \leq n$ there is a path $\varpi \preceq \pi$ with $\sigma_\xi = \sigma_\varpi$, such that for all locations $l$, where $\B \trans{\varpi} l$ and $\OG \definedon l$, 
  $\hat\vu_\xi \models \OG(l)  $ holds  .
  
\noindent
\textbf{Inductive Step}
We want to prove that for $|\omega'|=|\pi'|=n+1$ if $ \sigma_\omega' = \sigma_{\pi'^\S} 
    \quad\text{and}\quad 
    \vu_\omega' \models \eta^\S_{\pi'^\S}$ then $\sigma_\omega' = \sigma_{\pi'^\B} 
    \quad\text{and}\quad 
    \vu_\omega' \models \eta^\ini_{\pi'^\B},$ and for $\xi' \preceq \omega'$ we have a path $\varpi' \preceq \pi'$ with $\sigma_{\xi'}=\sigma_{\varpi'}$ and for all locations $l$,where $\B \trans{\varpi'}l$, $\hat\vu_{\xi'} \models \OG(l)$

Suppose $\omega'=\omega.(g\bar{w})$ and $\pi'^\S=\pi^S.(\alpha,\phi_\S,a_\S)$  and $\pi'^\B=\pi^\B.(\alpha,\phi_\B,a_\B)$ with $\alpha=(g,\ivbar)$, 
\begin{itemize}
    \item we know $\vu_\omega' \models \eta^\S_{\pi'^\S} $ so $\upshift\vu_{\omega} \sqcup \vu_{g\bar{w}}  \models \eta^\S_{\pi^\S}\upshift \land \phi_\S[a^\ini_{\pi^\S} \upshift] $
    \item [] From the induction hypothesis   $\vu_\omega \models \eta^\S_{\pi^\S}$ so $\vu_{g\bar{w}} \models \phi_\S[a^\ini_{\pi^\S} \upshift]$
    \item We want to show that $\vu_\omega' \models \eta^\ini_{\pi'^\B}$, lets also split this:
    \item $\upshift\vu_{\omega} \sqcup \vu_{g\bar{w}}  \models \eta^\B_{\pi^\B}\upshift \land \phi_\B[a^\ini_{\pi^\B} \upshift] $
    \item [] Again from the induction hypothesis   $\vu_\omega \models \eta^\B_{\pi^\B}$ so we only need to show $\vu_{g\bar{w}} \models \phi_\B[a^\ini_{\pi^\B} \upshift]$
    \item First  $a^\ini_{\pi^\B}=a^\ini_{\pi^\S} \upshift$ because from \autoref{def:bddtstosts} the assignments are equal for all switches and also $\phi_\S=\phi_\B$ or if $\B \trans{\pi.(\alpha,\phi_\B,a_\B)} l$ and $\OG \definedon l$ then $\phi_\S=\phi_\B \land \OG[\mapcv^g]$ which is stronger than $\phi_\B$
    \item [] in both cases,as $\vu_{g\bar{w}} \models \phi_\S[a^\ini_{\pi^\S} \upshift]$ we can conclude that $\vu_{g\bar{w}} \models \phi_\B[a^\ini_{\pi^\B} \upshift]$
    \item For the second part of the proof we should show $\hat\vu_{\xi'} \models \OG(l)$. By definition of concrete semantics we could write it as $\vu_{\xi'} \sqcup (\vu_{\xi'} \circ \mapcv^g) \models \OG(l)$ and as $\vu_{\xi'}=\upshift \vu_\xi \sqcup \vu_{g\bar{w}}$ we have $(\vu_\xi \sqcup \vu_{g\bar{w}}) \sqcup ((\vu_\xi \sqcup \vu_{g\bar{w}}) \circ \mapcv^g) \models \OG(l)$
    \item by induction hypothesis for all $\ell$ such that $\B \trans{\varpi} \ell$, $\vu_\xi \sqcup (\vu_\xi \circ \mapcv^g) \models \OG(\ell)$, so we only need to show $\vu_{g\bar{w}}\sqcup (\vu_{g\bar{w}} \circ \mapcv^g) \models \OG(l)$ for the location $l$ where $\B\trans{\varpi.(\alpha,\phi_B, a_B)} l$
    \item from previous item we had $\vu_{g\bar{w}} \models \phi_B \land \OG[\mapcv^g]$ by conjunction elimination $\vu_{g\bar{w}} \models \OG[\mapcv^g]$ and $\OG(l)= \OG[\mapcv^{g^{-1}}]$ we conclude that $\vu_{g\bar{w}}\sqcup (\vu_{g\bar{w}} \circ \mapcv^g) \models \OG(l)$ 
    
\end{itemize}
\textbf{From 2 to~1}

\noindent
  
\noindent
\begin{description}
\item[Base case.] Immediate, since then $\pi^\B=\pi^\S=\epsilon$, $l=\il$ for which $\OG\undefinedon l$, and all derived valuations are $\True$.

\item[Induction Hypothesis.] for the path with length $n$ we assume that when  $\sigma_\omega = \sigma_{\pi^\B} 
    \quad\text{and}\quad 
    \vu_\omega \models \eta^\ini_{\pi^\B},$ and  for every prefix $\xi \preceq \omega$ with length $|\xi| \leq n$ there is a path $\varpi \preceq \pi$ with $\sigma_\xi = \sigma_\varpi$, such that for all locations $l$, where $\B \trans{\varpi} l$ and $\OG \definedon l$, 
  $\hat\vu_\xi \models \OG(l)$ holds, then $\sigma_\omega = \sigma_{\pi^\S}$
     and  
    $\vu_\omega \models \eta^\S_{\pi^\S}$.

\item[Induction Step]
We want to prove that for $|\omega'|=|\pi'|=n+1$ if $ \sigma_\omega' = \sigma_{\pi'^\B} 
    \quad\text{and}\quad 
    $  $\sigma_\omega' = \sigma_{\pi'^\B} 
    \quad\text{and}\quad 
    \vu_\omega' \models \eta^\ini_{\pi'^\B},$ and for $\xi' \preceq \omega'$ we have a path $\varpi' \preceq \pi'$ with $\sigma_{\xi'}=\sigma_{\varpi'}$ and for all locations $l$,where $\B \trans{\varpi'}l$, $\hat\vu_\xi' \models \OG(l)$ then $ \vu_\omega' \models \eta^\S_{\pi'^\S}$
    
Suppose $\omega'=\omega.(g\bar{w})$  and $\pi'^\B=\pi^\B.(\alpha,\phi_\B,a_\B)$ with $\alpha=(g,\ivbar)$, and $\pi'^\S=\pi^S.(\alpha,\phi_\S,a_\S)$ 
\begin{itemize}
    \item $\upshift\vu_{\omega} \sqcup \vu_{g\bar{w}}  \models \eta^\B_{\pi^\B}\upshift \land \phi_\B[a^\ini_{\pi^\B} \upshift] $
    \item From induction hypothesis $\vu_{\omega} \models \eta^\B_{\pi^\B}$ then $\vu_{g\bar{w}}  \models \phi_\B[a^\ini_{\pi^\B} \upshift]  $
    \item  we want to show $\vu_\omega' \models \eta^\S_{\pi'^\S} $ splitting the statement: $\upshift\vu_{\omega} \sqcup \vu_{g\bar{w}}  \models \eta^\S_{\pi^\S}\upshift \land \phi_\S[a^\ini_{\pi^\S} \upshift] $
    \item From induction hypothesis $\vu_{\omega} \models \eta^\S_{\pi^\S}$ (both sides are upshifted here, which doesn't change the meaning, old values assigned to upshifted variables) then we only need  to show $\vu_{g\bar{w}}  \models \phi_\S[a^\ini_{\pi^\S} \upshift]  $
    \item By definition of \autoref{def:bddtstosts}, as assignments of all switches are the same in BDDTS and STS then $a^\ini_{\pi^\S} \upshift=a^\ini_{\pi^\B} \upshift$
    and again by the same definition $\phi_\S=\phi_\B$ for switches that don't end in a goal location and $\phi_\S=\phi_\B \land \OG[\mapcv^g] $ for switches ending in a goal location. For the former case we can directly conclude $\vu_{\omega'} \models \eta^\S_{\pi'^\S}$ but for the latter we need the second part of the assumption:
    \item if $l$ is the goal location after the path $\varpi'$ meaning:
    $\B \trans \varpi' l$ which is equal to  $\B \trans{ \varpi.(\alpha_B,\phi_B,a_B)} l$ from the assumption we know $\hat\vu_\xi' \models \OG(l)$, rewriting it as $(\vu_\xi \sqcup \vu_{g\bar{w}}) \sqcup ((\vu_\xi \sqcup \vu_{g\bar{w}}) \circ \mapcv^g) \models \OG(l)$
    \item by induction hypothesis for all $\ell$ such that $\B \trans{\varpi} \ell$, $\vu_\xi \sqcup (\vu_\xi \circ \mapcv^g) \models \OG(\ell)$, so up to $\omega$ this holds: ($\B \trans{\varpi} \ell$, $\vu_\omega \sqcup (\vu_\omega \circ \mapcv^g) \models \OG(\ell)$) and also as $\vu_{\xi'} \models \OG(l)$ from the assumption it holds that $\vu_{g\bar{w}}\sqcup (\vu_{g\bar{w}} \circ \mapcv^g) \models \OG(l)$ for the location $l$
    \item We had that $\vu_{g\bar{\omega}}\models \phi_\B[a^\ini_{\pi^\B} \upshift]$ and from previous item $\vu_{g\bar{w}}\sqcup (\vu_{g\bar{w}} \circ \mapcv^g) \models \OG(l)$ so we can conclude that $\vu_{g\bar{\omega}} \models \phi_\B[a^\ini_{\pi^\B} \upshift] \land \OG[\mapcv^g] $ which is equivalent to $\vu_{g\bar{\omega}} \models \phi_\S[a^\ini_{\pi^\S}]$ closing the proof.
\end{itemize}
\end{description}
\end{proof}

\begin{lemma}\label{lem:omega-vs-EC}
Let $\B$ be a saturated output-rich BDDTS and let $\ini \in \TU(\emptyset)^V$ such that $\sem\ini \models \IG_\B$, and consider the LTS $\L=\sem{\ISS\B\ini}$. For any gate value sequence $\omega \in \GUU^*$, $q_0\trans\omega$ in $\L$ if and only if $\vu_\omega \models \EC_{\B,\ini}(\sigma_\omega)$ and for all $\xi\preceq \omega$, $\hat\vu_\xi \models \GI_{\B,\ini}(\sigma_\xi)$.
\end{lemma}
\begin{proof}
Let $\S=\ISS{\B}{\ini}$.
\begin{description}
\item[If.] Assume $\vu_\omega \models \EC_{\B,\ini}(\sigma_\omega)$ and for all $\xi\preceq \omega$, $\hat\vu_\xi \models \GI_{\B,\ini}(\sigma_\xi)$. It follows that there is a path $\il\trans{\pi^\B}$ in $\B$ with $\sigma_\pi=\sigma_\omega$ and $\vu_\omega\models \eta^\ini_\pi$; moreover, for al $\varpi\preceq \pi^\B$ if $\il\trans\varpi l$ for $\OG\definedon l$ then for the corresponding $\xi\preceq\omega$ (i.e., such that $\sigma_\xi=\sigma_\varpi$) both $\vu_\xi\models \eta^\ini_\varpi$ and $\hat\vu\models \OG(l)$. \autoref{lem:eta-B-to-S} then implies that there is a path $\pi^\S$ such that $\il\trans{\pi^\S}$ in $\S$ such that $\omega\models \eta^\S_{\pi^\S}$, which (by \autoref{prop:eta-vs-vu}) implies $(\il,\sem\ini)\trans\omega$.

\item[Only if.] Assume $q_0\trans\omega$ in $\L$. By \autoref{prop:eta-vs-vu}, there is a path $\il\trans{\pi^\S}$ in $\S$ such that $\sigma_{\pi^\S}=\sigma_\omega$ and $\vu_\omega\models \eta^\S_{\pi^\S}$. \autoref{lem:eta-B-to-S} then implies the existence of a path $\pi^\B$ such that $\sigma_{\pi^\B}=\sigma_\omega$ and 
\begin{enumerate}
\item[(i)] $\vartheta_\omega\models \eta^\ini_{\pi^\B}$, implying $\vartheta_\omega\models \EC_{\B,\ini}(\sigma_\omega)$, and 
\item[(ii)] and for all $\xi\preceq \omega$, if $\varpi\preceq\pi^\B$ is the corresponding path prefix (meaning $\sigma_\varpi=\sigma_\xi$) and $\il\trans{\varpi} l$ with $\OG\definedon l$, then $\hat\vu_\xi\models \OG(l)$, implying $\hat\vu_\xi\models \GI_{\B,\ini}(\sigma)$. 
\end{enumerate} 
\end{description}
\end{proof}

\symbolicvsconcrete*

\begin{proof}[of \autoref{thm:sym-concrete-relation}] Let $\S=\ISS{\B}{\ini}$ and $\L=\sem\S$ and $\cal T=\TC(\L)$. We abbreviate $\EC_{\B,\ini}$ to $\EC$ and $\GI_{\B,\ini}$ to $\GI$.
\begin{enumerate}
\item 
\begin{description}
\item[If.] Assume that, for a given $\xi\preceq \omega$, $\vu_{\xi}\models \EC_{\B,\ini}(\sigma_\xi)$ and conditions \emph{(i)} and \emph{(ii)} hold. \autoref{lem:omega-vs-EC} then implies $q_0\trans\xi q$ for some $q$ and, moreover, $q_0\ntrans{\xi u}$ for all $u\in \GUU$; hence $q$ is a sink state in $\L$. If $q\in Q^\Cl$, the underlying BDDTS location $l$ must be closed, but then due to the fact that $\B$ is saturated there is an outgoing transition $l\trans{\alpha,\phi,a}$ with $g_\alpha\in G_o$ and satisfiable $\phi$, contradicting $\sink(q)$. It follows that $q\in Q^\Op$ (in $\L$) and hence $q\in P$ (in $\cal T$), implying $\omega\passes \cal T$.

\item[Only if.]
Assume $\omega\passes \cal T$ due to $q_0\trans\xi q\in P$ for $\xi\preceq \omega$. This transition sequence then also exists in $\L$ and, moreover, $q\in Q^\Op$ and $q\ntrans u$ for all $u\in \GUU$. \autoref{lem:omega-vs-EC} then implies $\vu_\xi\models \EC(\sigma_\xi)$ and for all $\xi'\preceq \xi$, $\hat\vu_{\xi'}\models \GI(\sigma_{\xi'})$ (hence condition~\emph{(i)} of \autoref{thm:sym-concrete-relation}.1 is satisfied), and also for all $u\in \GUU$ either $\vu_{\xi u}\not\models \EC(\sigma_{\xi u})$ or for some $\xi' \preceq \xi u$, $\hat\vu_{\xi'}\not\models \GI(\sigma_{\xi'})$. The only candidate for the latter is $\xi'=\xi u$; however, this would mean that there is an outgoing transition with an output gate from the location underlying $q$, and (due to output richness) this conflicts with $q\in Q^\Op$. We may conclude that condition~\emph{(ii)} of \autoref{thm:sym-concrete-relation}.1 is also satisfied.
\end{description}

\item
\begin{description}
\item[If.] Assume that, for a given $\xi\preceq \omega$, $\vu_{\xi}\models \EC(\sigma_\xi)$ and $\hat\vu_\xi \not\models \GI(\sigma_\xi)$. Let $u_1\cdots u_n\preceq \xi$ be the largest prefix such that $\hat\vu_{u_1\cdots u_i}\models \GI(\sigma_{u_1\cdots u_i})$ for all $i<n$ whereas $\hat\vu_{u_1\cdots u_n}\not\models \GI(\sigma_{u_1\cdots u_n})$. According to \autoref{lem:omega-vs-EC}, it follows that $q_0\trans{u_1\cdots u_{n-1}} q\ntrans{u_n}$. Because $\vu_{u_1\cdots u_n}\models \EC(\sigma_{u_1\cdots u_n})$ it follows that $\B$ has a transition $l\trans{\alpha,\phi,a} l'$ for the underlying location $l$ of $q$ such that $\vu_{u_1\cdots u_n}\models \phi$; since nevertheless $q\ntrans{u_n}$ it must be the case that $\OG\definedon l'$. Because $\B$ is output-rich, this implies $g_\alpha\in G_o$ and $l\in L^\Cl$, hence $u_n\in A_o$ and $q\in Q^\Cl$; therefore $q\trans{u_n} q_\f$ in $\cal T$. We may conclude $\cal T\fails \omega$.

\item[Only if.] Assume that $\omega\fails \cal T$ due to $q_0\trans{\xi u} q_\f$ for $\xi u\preceq \omega$. It follows that $q_0\trans\xi q\ntrans u$ in $\L$ for some $q\in Q^\Cl$ and $u\in A_o$. According to \autoref{lem:omega-vs-EC}, it follows that $\vu_\xi\models \EC(\sigma_\xi)$ and for all $\xi'\preceq \xi$, $\hat\vu_{\xi'}\models \GI(\sigma_{\xi'})$, whereas either $\vu_{\xi u}\not\models \EC(\sigma_{\xi u})$ or $\hat\vu_{\xi u}\not\models \GI(\sigma_{\xi u})$. However, the first contradicts the saturation of $\B$ (noting that the underlying location of $q$ is closed and $g_u$ is an output action). In conclusion, $\xi u$ satisfies the condition of \autoref{thm:sym-concrete-relation}.2.
\end{description}
\end{enumerate}
\end{proof}

\corollarytestingequiv*
\begin{proof}
We prove the left-to-right implications; the inverse direction follows by symmetry.
\begin{itemize}
\item Assume $\TC(\B_1,\ini)\fails\omega$ for some $\B_1\in\BB_1\restr\ini$; then (due to \autoref{thm:sym-concrete-relation}.2) there is a $\xi\prec \omega$ such that $\vu_\xi\models \EC_{\B_1,\ini}(\sigma_\xi)$ and $\hat\vu_\xi \not\models \GI_{\B_1,\ini}(\sigma_{\xi})$.
Due to $\BB_1\simeq \BB_2$ (disjunctive equivalence of $\EC$), there is a $\B_2\in \BB_2\restr\ini$ such that $\vu_\xi\models \EC_{\BB_2,\ini}(\sigma_\xi)$, and (conjunctive equivalence of $\GI$) also $\hat\vu_\xi\not\models \EC_{\B_2,\ini}$. 

\item By contraposition. Assume $\TC(\B_2,\ini)\not\passes\omega$ for some $\B_2\in\BB_2\restr\ini$; then either there is a $\xi\preceq \omega$ such that $\vu_\xi\models \EC_{\B_2,\ini}(\sigma_\xi)$ and $\hat\vu_\xi\not \models \GI_{\B_2,\ini}(\sigma)$, or there is some $\xi\succ \omega$ such that $\vu_\xi\models \EC_{\B_2,\ini}(\sigma_\xi)$. In the first case we have $\omega\fails \TC(\B_2,\ini)$, which by the first clause of this proof implies $\omega\fails\TC(\B_1,\ini)$ and hence $\omega\not\passes \TC(\B_1,\ini)$ for some $\B_1\in\BB_1\restr\ini$. In the second case, due to $\BB_1\simeq \BB_2$ there is a $\B_1\in \BB_1\restr\ini$ such that $\vu_\xi\models \EC_{\BB_1,\ini}(\sigma_\xi)$, again implying $\omega\not\passes \TC(\B_1,\ini)$.
\end{itemize}
\end{proof}


\end{document}